\documentclass[11pt]{article}

\usepackage[latin1]{inputenc}
\usepackage[T1]{fontenc}
\usepackage[english]{babel}
\usepackage{amsfonts}
\usepackage{amssymb}
\usepackage{amsmath}
\usepackage{amsthm}
\usepackage{authblk}

\def\bc{\begin{center}}
\def\ec{\end{center}}
\def\be{\begin{equation}}
\def\ee{\end{equation}}
\def\bea{\begin{eqnarray}}
\def\eea{\end{eqnarray}}
\def\easp{\mathbb{E}}
\newtheorem{proposition}{Proposition}
\newtheorem{lemma}{Lemma}
\newtheorem{theorem}{Theorem}

\newtheorem{remark}{Remark}
\newtheorem{definition}{Definition}

\begin{document}

\title{Notes on the p-spin glass studied via Hamilton-Jacobi and  Smooth-Cavity  techniques}

\author[1]{Elena Agliari}
\author[2]{Adriano Barra}
\author[1]{Raffaella Burioni}
\author[1]{Aldo Di Biasio}
\affil[1]{Dipartimento di Fisica, Universit\`a di Parma \newline and INFN, Gruppo collegato di Parma}
\affil[2]{Dipartimento di Fisica, Sapienza Universit\`a di Roma}

\maketitle

\begin{abstract}

In these notes, we continue our investigation of classical toy models of disordered statistical mechanics, through techniques recently developed and tested mainly on the paradigmatic Sherrington-Kirkpatrick spin glass.  Here, we consider the p-spin-glass model with Ising spins and interactions drawn from a normal distribution $\mathcal{N}[0,1]$. After a general presentation of its properties (e.g. self-averaging of the free energy, existence of a suitable thermodynamic limit), we study its equilibrium behavior within the Hamilton-Jacobi framework and the smooth cavity approach. Through the former we find both the RS and the 1-RSB expressions for the free-energy, coupled with their self-consistent relations for the overlaps. Through the latter, we recover these results as irreducible expression, and we study the generalization of the overlap polynomial identities suitable for this model;  a discussion on their deep connection with the structure of the internal energy and the entropy closes the investigation.

\end{abstract}


\section{Introduction}

In these notes we continue our investigation on the mathematical methods and the physics underlying many body interactions, namely we adapt  recent mathematical techniques to the study of equilibrium statistical mechanics of p-spin glasses. In the past we analyzed p-spin systems with the simpler ferromagnetic couplings \cite{barrapspinrem} and p-spin systems with diluted coupling \cite{boh}, while now we turn to p-spin systems with frustrated couplings, which are termed p-spin glasses
\cite{mezardgross, gardner}.

We first introduce the model, with all the necessary definitions stemmed from statistical mechanics, and then we adapt the Hamilton-Jacobi technique (developed for the Sherrington-Kirkpatrick model by Guerra \cite{guerrahj} and later enlarged to a broad validity \cite{barrainterpol, tesispec, barragenmech,starr}) to these systems, so to be able to solve the model (in some physical approximation, that is, replica symmetric and one-step of broken replica symmetry, as discussed later), without any relation with the  original statistical mechanics framework.

This has two advantages: the development of a clear and powerful mathematical alternative to solve the thermodynamics of these many body systems, and a further rigorous confirmation of results raised in the theoretical physics scenario.

Then, we adapt the method of the smooth cavity to the same problem to obtain another series of results: in particular, after recovering a clear picture of the thermodynamics in perfect agreement with the previous part of the work and with existing results, we focus on the polynomial identities often called Aizenman-Contucci \cite{aizcont} and Ghirlanda-Guerra \cite{gg} relations. We will show how to prove their validity  even for the p-spin glasses considered here and we will try to revise their deep physical meaning ultimately offering a unifying framework where cavity fields \cite{mpv} and stochastic stability \cite{contuccistocstab} merge to work synergically \cite{barradesanctis}. Furthermore, comparison among the results obtained with both the methods will provide the reader with a deeper understanding of the techniques themselves as well as of the physical properties of these models.

In order to be comprehensible for both the communities of theoretical physicists and of mathematical physicists, the two methods are exposed with a slightly different approach. In the former (closer to the first community), results are presented in form of a theorem following the related proof, which is never explicitly expressed as a "proof", while in the latter (closer to the second community) results are first declared and then proved.
\newline
Finally, in the last section we discuss results and possible outlooks.

\section{The model and the related statistical mechanics package}

The p-spin glass is the model for a system of spins $\sigma$, i. e. dichotomic variables which can take the values $\pm 1$,
interacting together in $p$-tuples with random couplings $J_{i_1.... i_p}$, and, possibly, with an external field $h$. The Hamiltonian is the function which defines the model
and physically speaking represents the extensive energy associated with a given configuration of the spins, for a certain value of the couplings and of the external field.
\begin{definition}
Given a system of $N$ spins $\sigma_i$, $i=1,...,N$,
the Hamiltonian associated with a configuration $\sigma= \{ \sigma_1,...,\sigma_N \}$ of the spins,
interacting in $p$-tuples and with an external uniform magnetic field $h$, is defined as follows:
\be
\label{eq:hamiltoniana}
H_N(\sigma, J, h) = -\sqrt{\frac{p!}{2 N^{p-1}}} \sum_{i_1<...<i_p}^{1,N} J_{i_1...i_p} \sigma_{i_1}...\sigma_{i_p} - h\sum_{i=1}^N \sigma_i.
\ee
\end{definition}
The first summation is taken over all the possible choices of indices $1\leq i_1 <...< i_p \leq N$
and the couplings $J$ are independent standard Gaussian random variables.
This can be considered as a generalization of the well known Sherrington-Kirkpatrick model (SK) and its interest lays in the fact that its low temperature behavior is much simpler than in the SK model.
The normalization factor preceding the first sum ensures that the Hamiltonian is an extensive quantity
(i.e. proportional to the number of spins $N$) and
the $2$ at the denominator allows recovering the SK definition when $p=2$.
\newline
For the sake of simplicity we only consider the case of an even number $p$ of interacting spins.
In this case the system has a gauge symmetry when the external field $h$ is set equal to zero:
it is left invariant under the transformation $\sigma_{i_k} \to \sigma_{i_k} \sigma_{i_{p+1}}$ for all $k=1, 2, ..., p$.
Moreover, we assume that the external field vanishes, thus we neglect the second term:
in fact, this is a one-body term, which is simple to deal with.
In the following, $H_N(\sigma, J)$ has to be interpreted as $H_N(\sigma, J, 0)$.

For this model, the investigation of the free energy and its decomposition via Hamilton-Jacobi technique or in
terms of a cavity function and the energy can still be performed, but the simple mathematical treatment of the SK,
ultimately due to the second order nature of its phase transition allowing
expansions in small overlaps, is lost whenever $p > 2$ because the transition becomes first order.

This is an interesting remark because, when using the replica trick, the p-spin
models are always thought of as simpler cases. This has a deep physical counter-part:
the covariance of the Hamiltonian is given by the overlap to the power
$p$, so, for example the SK Hamiltonian has covariance $\sim N q^2$, while a generic
p-spin model has a covariance $\sim N q^p$.
Of course, as the overlap is bounded by
one, this means that by increasing the order of interactions $p$, these correlations
become more and more negligible until, in the limit $p \to \infty$,  one recovers an uncorrelated model,
i. e. the Random Energy Model \cite{derridarem}. The latter is analytically solvable without either
replica tricks or cavity field techniques.

Through a direct calculation (by applying Wick theorem), we can check that the normalization of the Hamiltonian ensures
a correct volume scaling for the energy such that $$\lim_{N \to \infty}\langle - H_N (\sigma,J)/N \rangle \leq c \in \mathbb{R}.$$
\newline
All physical information is encoded in the free energy density $f(\beta) = \lim_{N \to \infty}f_N(\beta)$.
\begin{definition}
The free energy density $f_N(\beta)$ at finite volume $N$, which is a function of the inverse temperature $\beta=1/T$, is defined as
\be
\label{eq:energialibera}
f_N (\beta) \equiv - \frac{1}{\beta N} \mathbb{E} \log Z_N(\beta, J)
\equiv - \frac{1}{\beta N} \mathbb{E} \log \sum_{\sigma} e^{-\beta H_N(\sigma, J)},
\ee
\end{definition}
where $Z_N$ is called the partition function
and $\mathbb{E}$ stands for the expected value with respect to all the $J$'s.
As usual, the sum is over the $2^N$ configurations $\sigma = \{ \sigma_1, \sigma_2,...,\sigma_N \}$ of the spins.
Sometimes it is more convenient to deal with the "pressure" $$\alpha(\beta) = \lim_{N \to \infty} \alpha_N(\beta) = \lim_{N \to \infty} -\beta f_N(\beta).$$
These are the so-called \emph{quenched} free energy/pressure, where the disorder is "frozen"
and which are more difficult to compute than the \emph{annealed} ones,
where the expectation is taken directly in the partition function.
Using the property $\mathbb{E} \exp \lambda z = \exp \lambda^2/2 $ valid for a standard random variable $z$, the computation of
the \emph{annealed} free energy density $ f_A(\beta)$ is in fact straightforward:
\begin{lemma}
The annealed free energy density is given by
\be\label{annoso}
-\beta f_{A}(\beta) \equiv \lim_{N \to \infty}  \frac{1}{ N}  \log \mathbb{E} Z_N(\beta, J) =  \log 2 + \beta^2 / 4.
\ee
\end{lemma}
We notice that when $\beta$ is sufficiently small, namely at high temperature, this result coincides with that obtained 
for the quenched average.
Physically speaking, when the temperature is high enough, spins are no longer correlated and averaging the disorder
directly in the partition function (which in some way means that it participates to thermodynamic equilibrium) and then taking the logarithm
is the same as averaging $\log Z_N$.

\begin{definition}
If $F(\sigma)$ is a (real-valued) physical observable, we denote the Boltzmann average with
\be
\omega(F(\sigma)) = (1/Z_N(\beta,J)) \sum_{\sigma} F(\sigma) \exp (-\beta H_N(\sigma, J)).
\ee
\end{definition}
This can be generalized by considering two or more independent replicas of the system with the same disorder,
so that if $F(\sigma, \sigma')$ is an observable depending on the configuration of two replicas $\sigma, \sigma'$,
 its Boltzmann average is  $\Omega(F(\sigma, \sigma')) \equiv (1/Z^2_N(\beta, J)) \sum_{\sigma} \sum_{\sigma'} F(\sigma, \sigma') \exp (-\beta H(\sigma)-\beta H(\sigma'))$.
Notice that, even if we did not write it explicitly, $\omega$ depends on the disorder $J$, too.
We denote the average over the disorder with brackets: $\langle F(\sigma, \sigma') \rangle \equiv \easp \,\Omega(F(\sigma, \sigma')) $.


\section{Thermodynamic limit}

 The quantity one is typically interested in is actually the thermodynamical limit of the quenched free energy
 \be
 f(\beta) = \lim_{N \to \infty} f_N(\beta).
\ee
Guerra and Toninelli first were able to find out a mathematical strategy to prove the existence of the thermodynamic limit for these frustrated systems \cite{guerratoninelli,guerratoninelli2}, which, for the sake of completeness, we briefly outline:
\begin{theorem}
The thermodynamic limit of the free energy density exists and it is equal to its infimum
\be
 \lim_{N \to \infty} f_N(\beta) = \inf_{N } \left( -\frac{1}{\beta N} \mathbb{E} \log Z_N(\beta,J) \right).
\ee
\end{theorem}
\begin{proof}
Let us consider two separated systems, one constituted by $N$ elements and the other one
by two independent subsystems (labeled by $1$ and $2$) with $N = N_1 + N_2$ elements.
The Hamiltonian and free energy density for the first system correspond to expressions (\ref{eq:hamiltoniana}), (\ref{eq:energialibera}),
while for the second system, indicating with $\sigma^{(1)}$ and $\sigma^{(2)}$ the two subsets $\{\sigma_1,...,\sigma_{N_1}\}$ and $\{\sigma_{N_1+1},...,\sigma_{N}\}$,
we have an Hamiltonian
\bea
 H_{N_1}(\sigma^{(1)}, J') + H_{N_2}(\sigma^{(2)},J'') &=&  -\sqrt{\frac{p!}{2 N_1^{p-1}}} \sum_{1\leq i_1<...<i_p \leq N_1} J_{i_1...i_p}'\sigma_{i_1}...\sigma_{i_p} \\
  && - \sqrt{\frac{p!}{2 N_2^{p-1}}} \sum_{N_1< i_1<...<i_p \leq N} J_{i_1...i_p}'' \sigma_{i_1}...\sigma_{i_p},
\eea
where the $J'$ and $J''$ are distributed as the $J$,
and an extensive free energy given by
\be
\mathbb{E} \log \sum_{\sigma^{(1)}} \exp(-\beta H_{N_1}(\sigma^{(1)}, J')) 
+\mathbb{E} \log \sum_{\sigma^{(2)}}\exp(-\beta H_{N_2}(\sigma^{(2)},J'')).
\ee
Let us introduce a new fundamental quantity, called overlap, which measures
the correspondence between two configurations of spins belonging to different replicas of the system
\begin{definition}
The overlap $q_{\sigma \sigma'}$ between two configurations $\sigma$ and $\sigma'$ is defined as
\be
q_{\sigma \sigma'} \equiv \frac{1}{N}\sum_{i=1}^N\sigma_i \sigma_i'.
\ee
\end{definition}
In the same way, we define overlaps for the two subsystems $1$ and $2$ making up the second system as
\bea
q_{\sigma \sigma'}^{(1)} &=& \frac{1}{N_1}\sum_{i=1}^{N_1}\sigma_i \sigma_i',  \\
q_{\sigma \sigma'}^{(2)} &=& \frac{1}{N_2}\sum_{i=N_1+1}^N\sigma_i \sigma_i'.
\eea
Choosing a proper free energy, which for $t \in [0,1]$ interpolates between the free energies of the two systems presented before, 
\be
\frac{1}{N}\mathbb{E}\log Z_N(t) = \frac{1}{N}\mathbb{E}\log \sum_{\sigma} \exp \left[ \beta \sqrt{t} H_N(\sigma, J) + \beta \sqrt{1-t}(H_{N_1}(\sigma^{(1)}, J')+H_{N_2}(\sigma^{(2)},J')  )\right],
\ee
we can easily compute its derivative with respect to the parameter $t$:
\be
\frac{d}{dt} \frac{1}{N}\mathbb{E}\log Z_N(t) = -\frac{\beta^2}{4}\left(\langle q_{12}^p \rangle_t - \frac{N_1}{N}\langle (q_{12}^{(1)})^p \rangle_t - \frac{N_2}{N}\langle (q_{12}^{(2)})^p \rangle_t \right),
\ee
where $\langle .\rangle_t$ is the average over all the disorder $J,  J', J''$ of the generalized interpolating Boltzmann state.
Since the function $q \to q^p$ is convex for even $p$, and
\be
\label{eq:covsub}
q_{\sigma \sigma'} = \frac{N_1}{N} q_{\sigma \sigma'}^{(1)} + \frac{N_2}{N} q_{\sigma \sigma'}^{(2)}
\ee
the derivative of the interpolating free energy is always non-negative
\be
\frac{d}{dt} \frac{1}{N}\mathbb{E}\log Z_N(t) \geq 0.
\ee
Integrating this equation between $0$ and $1$, it is straightforward to show that the thermodynamic pressure is superadditive:
$$N\alpha_N(\beta) \geq N_1 \alpha_{N_1}(\beta) + N_2 \alpha_{N_2}(\beta).$$
Hence, being $\alpha(\beta)= - \beta f(\beta)$, the quenched free energy is sub-additive in the system size.
By noticing that it is also limited, e.g. by its annealed value $\alpha(\beta) \leq \log 2 + \beta^2/4$ (see eq.(\ref{annoso})), the existence of its thermodynamic limit is shown, mirroring the original scheme by Guerra and Toninelli \cite{guerratoninelli}.
\end{proof}

\section{First approach: The Hamilton-Jacobi technique}

We now consider the analogy between the p-spin glass model
and the a proper mechanical system, obeying a certain Hamilton-Jacobi equation.
Interestingly, the potential in this equation is related to the fluctuations of the order parameter for the
corresponding thermodynamic system.
As we will see, neglecting this potential we will be able to reconstruct the free energy density
for the original model.

To this aim, let us consider the interpolating partition function, depending on the non-negative parameters $ t $ and $x$ (which symbolically may be thought of as a fictitious space-time continuum):
\be\label{bolzo}
Z_N(t,x) = \sum_{\sigma} \exp \left( \sqrt{\frac{t p!}{2 N^{p-1}}} \sum_{i_1<...<i_p}^{1,N} J_{i_1...i_p} \sigma_{i_1}...\sigma_{i_p}
+ \sqrt{x} \left(\frac{p}{2} Q^{p-2}(\beta) \right)^{1/4} \sum_i J_i \sigma_i \right).
\ee
The $J_i$'s are independent random variables, with the same distribution as the $J_{i_1...i_p}$ , and
represent an external random field, while $Q(\beta)$ is a regular function of $\beta$,
which we will later identify with the average overlap between two replicas of the system endowed with the same disorder.
Note that we omitted to write explicitly the dependence of  $Z_N$ on $\beta$ and on the $J$'s and we will refer to the free energy both at finite size and in the infinite volume limit when there is no danger of confusion.

We may consequently define an interpolating free energy as
\begin{definition}
The interpolating free energy density is defined as
\be
\label{eq:alfainterpol}
\alpha_N(t,x) \equiv \frac{1}{N} \easp \log Z_N(t,x),
\ee
\end{definition}
where the expectation $\easp$ is taken with respect to all the $J$'s, that is with respect to the mutual interactions between spins as well as on
the external random fields. It is immediate to see that the true physical free energy is obtained by taking $t=\beta^2$ and $x=0$,
so our strategy will consist in computing the interpolating free energy (\ref{eq:alfainterpol}) and obtaining the statistical mechanics by choosing
the right values of the parameters $t, x$.

We may now proceed to compute the derivatives of $\alpha$ with respect to the parameters.
With an integration by parts, and neglecting terms which are unimportant in the thermodynamic limit, we obtain
the following
\begin{lemma}
The derivatives of $\alpha(t,x)$ with respect to the parameters $t,x$ are
\bea
\label{eq:dtalfa0}
\partial_t \alpha(t,x) & = & \frac{1}{4} \left(1- \langle q_{\sigma \sigma'}^p \rangle_{t,x} \right),\\
\label{eq:dxalfa0}
\partial_x \alpha(t,x) & = & \frac{1}{2}  \left(\frac{p}{2}Q^{p-2}(\beta) \right)^{1/2} \left(1- \langle q_{\sigma \sigma'} \rangle_{t,x} \right),
\eea
\end{lemma}
where the generalized brackets $\langle . \rangle_{t,x}$ are meant to weight the observable with the generalized Boltzmann factor implicitly defined in eq. (\ref{bolzo}).

We then define a new function, which will play the role
of the bridge with a "mechanical" description.
\begin{definition}
The Hamilton principal function $S(t,x)$ is defined as
\be
\label{eq:defs}
S(t,x) \equiv 2 \alpha(t,x) - x \left[\frac{p}{2} Q^{p-2}(\beta)\right]^{1/2} -\frac{t}{2}\left[1+ \left(\frac{p}{2} -1\right) Q^p(\beta) \right].
\ee
\end{definition}
The derivatives of $S(t,x)$ are immediately deduced by  (\ref{eq:dtalfa0}) and (\ref{eq:dxalfa0}):
\bea
\partial_t S(t,x) & = & -\frac{1}{2} \langle q_{\sigma \sigma'}^p \rangle_{t,x} -\left( \frac{p}{4}- \frac{1}{2} \right) Q^p(\beta), \\
\partial_x S(t,x) & = & -  \left(\frac{p}{2} Q^{p-2}(\beta) \right)^{1/2}  \langle q_{\sigma \sigma'} \rangle_{t,x}.
\eea
Lastly, we introduce a proper potential.
\begin{definition}
The potential V(t, x) for the mechanical problem is defined as
\be\label{potente}
V(t,x) \equiv
 \frac{1}{2} \left(\langle q_{\sigma \sigma'}^p \rangle_{t,x} - Q^{p}(\beta) \right)
 + \frac{p}{4} \left( Q^{p}(\beta) - Q^{p-2}(\beta) \langle q_{\sigma \sigma'} \rangle_{t,x}^2 \right).
\ee
\end{definition}
With these definitions we are now able to formulate our problem (solving the thermodynamics of the p-spin model) as a suitable mechanical model.
\begin{proposition}
The Hamilton principal function $S(t, x)$, together with the potential $V(t, x)$ satisfies the
Hamilton-Jacobi equation:
\be\label{problema}
\partial_t S(t,x) + \frac{1}{2} (\partial_x S(t,x))^2+V(t,x) = 0.
\ee
\end{proposition}
Now we assume that the variance of the generalized overlap vanishes
\be
\langle q_{\sigma \sigma'}^2 \rangle_{t, x} = \langle q_{\sigma \sigma'} \rangle_{t, x}^2
\ee
and make the identification
\be
\label{eq:qQ}
\langle q_{\sigma \sigma'} \rangle_{t,x} = Q(\beta).
\ee
These assumptions are very important as they imply, in statistical mechanics, the self-averaging property for the order parameter. Despite we assume them and not prove them, we simply note that, in order to keep finite the potential $V(t,x)$, even in the $p \to \infty$ limit (which is the interesting case of the random energy model), the expression in the brackets of the second term at the  r.h.s. of eq. (\ref{potente}) must vanish, hence recovering our assumption.
Under these hypotheses,  within the mechanical analogy we are developing,  the two terms of the potential $V(t, x)$ vanish allowing the system to a free motion,
and the corresponding solution $\bar{S}(t,x)$ is related
to the so-called replica-symmetric (RS) free-energy, which is the approximation of the free energy density $f_N(\beta)$ obtained by neglecting overlap fluctuations.

This phenomenology, as it is based on free-field propagation, gives straight lines as equations of motion:
\be
\label{eq:moto}
x(t) = x_0 - \left(\frac{p}{2} Q^{p}(\beta)\right)^{\frac{1}{2}} t,
\ee
where $x_0$ is the starting point.
When $x=0$ and $t=\beta^2$ (namely, in the point recovering the standard statistical mechanics framework) we get
\be
\label{eq:x0}
x_0 = \beta^2 \left( \frac{p}{2} Q^{p}(\beta) \right)^{1/2}.
\ee
The trajectories (\ref{eq:moto}) do not intersect, as stated in the following
theorem.
\begin{theorem}
Given a generic point $(x,t)$ with $x \geq 0$, $t \geq 0$, there exists a unique $x_0(x,t)$ such that
\be
x = x_0(x,t) - \langle q_{\sigma \sigma'} \rangle_{0, x_0(x,t)} \, t,
\ee
and a unique $\bar{q}(x,t) = \langle q_{\sigma \sigma'} \rangle_{0, x_0(x,t)}$ such that
\be
\bar{q}(x,t) = \int \frac{dz}{\sqrt{2 \pi}} e^{-z^2/2} \tanh^2 \left[  z \left(\frac{p}{2} Q^{p-2} \right)^{1/4} \sqrt{ x + \bar{q}(x,t) t} \right].
\ee
\end{theorem}
The proof is based on the fact that the point $t(x_0)$ at which the free trajectory intersects the $t$-axis is
a monotonous function of the starting point $x_0$, and can be found in \cite{guerrahj},
where the SK case is studied in detail.

The Hamilton-Jacobi equation admits both an Hamiltonian $\emph{H}(t,x)$ and a Lagrangian  $\emph{L}(t,x)$ description, being respectively
\begin{eqnarray}
\emph{H}(t,x) &=& \frac12 \left(\frac{d S(t,x)}{dx}\right)^2 + V(t,x), \\
\emph{L}(t,x) &=& \frac12 \left(\frac{d S(t,x)}{dx}\right)^2 - V(t,x).
\end{eqnarray}
As we are working in the assumption of zero potential, they both correspond to the kinetic energy only:
\begin{definition}
The kinetic energy $T(t, x)$ is given by
\be
T(t, x) \equiv \frac{1}{2}\left( \partial_x S(t, x) \right)^2 = \frac{p}{4} Q^p(\beta).\ee
\end{definition}
This definition allows the following proposition:
\begin{proposition}
The solution $\bar{S}(t,x)$ of the Hamilton-Jacobi problem (\ref{problema}) for $V(t, x)=0$ is
obtained by taking the function $S(t, x)$ in one point (e.g. at time $t=0$ and space $x=x_0$) and adding
the Lagrangian times $t$ (strictly speaking it should be times $(t-t_0)$ but we choose $t_0=0$).
\be
\label{eq:soluzionehjs}
\bar{S}(t,x)= S(0,x_0) + \emph{L}(t,x)t = S(0,x_0) + T(t,x)t.
\ee
\end{proposition}
\begin{remark}
The freedom in the assignation of the Cauchy problem plays an important role as, by choosing $t_0=0$, we are left with a one-body problem in the calculation of the starting point and all the technical difficulties are left in the propagator which, at the replica symmetric level (e.g. $V(t,x)=0$), simply reduces to the kinetic energy times time.
\end{remark}

From (\ref{eq:soluzionehjs}) and (\ref{eq:defs}), we obtain the corresponding expression for the generalized free energy $\bar{\alpha}(t,x)$ in the replica symmetric approximation (RS).
\be
\label{eq:generalizedalfa}
\bar{\alpha}(t,x) = \alpha(0,x_0) - \frac{1}{2} x_0 \left(\frac{p}{2} Q^{p-2}(\beta)\right)^{1/2} + \frac{p}{8} Q^p(\beta)  t
+ \frac{1}{2} x \, Q^{\frac{p-2}{2}}(\beta) + \frac{t}{4} \left[1+\left(\frac{p}{2}-1\right) Q^p(\beta) \right] .
\ee
Now it is easy to obtain the physical  free energy, since the free energy for $t=0$ does not contain the interaction and
may be computed straightforwardly
\be
\alpha(0, x_0) = \log 2 + \int \frac{dz}{\sqrt{2 \pi}} e^{-z^2/2} \log \cosh \left[\left(\frac{p}{2}Q^{p-2}(\beta) \right)^{\frac{1}{4}} \sqrt{x_0} z, \right],
\ee
so that, using (\ref{eq:x0}), we finally find the expression for the physical (RS) free energy as stated by the next theorem.
\begin{theorem}
The replica symmetric free energy $\bar{\alpha}(\beta)$ of the p-spin model, obtained under the assumption of zero potential $V(t,x)$ in the mechanical analogy, is encoded in the following formula (which must be extremized over the order-parameter):
\be\label{solRS}
\bar{\alpha}(\beta) = \log 2 + \int \frac{dz}{\sqrt{2 \pi}} e^{-z^2/2} \log \cosh \left[\beta \left(\frac{p}{2}Q^{p-1}(\beta)\right)^{\frac{1}{2}}  z \right]
+ \frac{\beta^2}{4} \left[1 + (p-1) Q^p(\beta) - p Q^{p-1}(\beta) \right].
\ee
\end{theorem}
This represents the RS free energy, which corresponds to the true free energy only for sufficiently small values of $\beta$ \cite{mpv}.
In fact, assuming a vanishing potential corresponds to neglect overlap fluctuations, and the overlap may be
identified with a single value (RS approximation) only for high temperatures.
\begin{proposition}
Dealing with the overlap, which is related to the initial velocity of the mechanical system, we obtain the following viscous Burger equation which encodes the standard self-consistency procedure of the statistical mechanics counterpart
\be
\label{eq:selfRS}
\langle q_{\sigma \sigma'} \rangle_{0, x_0} = Q(\beta) = \int \frac{dz}{\sqrt{2 \pi}} e^{-z^2/2} \tanh^2 \left[\beta \left(\frac{p}{2}Q^{p-1}(\beta)\right)^{1/2}  z \right].
\ee
\end{proposition}
Note that the correct SK Replica Symmetric free energy and self-consistence equation are recovered for $p=2$, and
both equations predict in this case a phase transition for $\beta=\beta_c=1$. Above this value the
Replica Symmetric solution ceases to be valid \cite{gardner}.

It will be useful for a comparison among results gained within this technique and the next one, to have a polynomial expansion through $Q(\beta)$ of the expression (\ref{solRS}), hence getting
\be\label{confronto}
\bar{\alpha}(\beta) \sim \log 2 + \frac{\beta^2}{4} + \frac{\beta^2}{4}(p-1)Q^p(\beta) - \frac{\beta^4}{8}p Q^{2(p-1)} + O(Q^{2(p-1)}).
\ee

\subsection{Extension to the Broken Replica Symmetry scenario}

We now extend the technique presented before to the case of one step of broken replica symmetry,
which is known to broaden the correctness of the solution to values of $\beta$ higher than those required by the previous approximation \cite{gardner}.
In general, it is possible to consider even several steps of broken symmetry, and in fact in the case of the SK model
the free energy for $\beta > \beta_c=1$ is obtained in the limit of infinite iterative steps
(this is the so-called full RSB or $\infty$-RSB scheme \cite{mpv}).
For higher $\beta$ a broken replica phase is the correct solution even in the case of $p>2$, so we want to investigate deeply even the mathematical architecture beyond the preserved replica symmetry.
Following the approach of \cite{guerrahj, tesispec}, we see that in order to account for breaking of this symmetry in our mechanical analogy, we have to enlarge our fictitious space-time by one extra spatial dimension for each step of replica symmetry breaking that we want to consider.
\newline
To this task, let us introduce the recursive generalized partition function $\tilde{Z}_N (t; x_1,...,x_K)$, depending on the non-negative real parameters $t$ and $x_1, ..., x_K$:
\be
\tilde{Z}_N (t; x_1,...,x_K) \equiv \sum_{\sigma} \exp \left[
\sqrt{ \frac{tp!}{2N^{p-1}} } \sum_{1\leq i_1<...<i_p\leq N} J_{i_1,...i_p} \sigma_{i_1}...\sigma_{i_p}
+ \sum_{a=1}^K \sqrt{x_a} \left( \frac{p}{2} Q_a^{p-2} \right)^{1/4} \sum_{i=1}^N J_i^a \sigma_i \right].
\ee
Here, as before, the $J_i^a$ are independent random Gaussian variables with zero mean and unitary variance,
and we denote by $\easp_a$ the expectation with respect to all the $J_i^a$ for $i=1,...N$.
The $Q_a(\beta)$ are regular functions of $\beta$ which may be identified with the values around which the overlap distribution accumulates,
and they are ordered in the interval $[0,1]$:
\be
0 \equiv Q_0(\beta) < Q_1(\beta)<...<Q_K(\beta) <1.
\ee

We denote the Boltzmann-Gibbs state associated to this partition function with $\tilde{\omega}(.)$,
and observe that the physical model is recovered by choosing $t=\beta^2$ and $x_a=0$ for $a=1,...,K$.

Given the  $K+1$ ordered real numbers within the interval $[0,1]$,
the typical nested structure of the broken replica symmetry is encoded in the generalized partition functions $Z_a$, defined recursively as
\be
\label{eq:defza}
Z_a = \left( \easp_{a+1} Z_{a+1}^{m_{a+1}} \right)^{1/m_{a+1}},
\ee
with $Z_K \equiv \tilde{Z}_N$ and $Z_0 \equiv \exp (\easp_1 \log Z_1)$. Note that this last definition is
obtained by the general one (\ref{eq:defza}) in the limit of $m_1 \to 0$.
The number $K$ of parameters $x_a$ (dimensions of our fictitious space-time) are then related to the number of steps of broken symmetry.
\newline
It is useful to define the quantities
\be
f_a \equiv \frac{ Z_a^{m_a} }{ \easp_a Z_a^{m_a} },
\ee
which are all non-negative and not greater than one,
and share with the $Z_a$ the property of depending
on the random fields $J_i^b$ only with $b \leq a$.

With these definitions we are now able to introduce the new states.
\begin{definition}
The generalized Boltmann-Gibbs states are defined as
\bea
\omega_a (.) & \equiv & \easp_{a+1}...\easp_{K}( f_{a+1} ... f_K \tilde{\omega}(.)), \\
\omega_K (.) & \equiv & \tilde{\omega}(.).
\eea
\end{definition}
Again, it is possible to define Boltmann-Gibbs states $\Omega_a$ for replicas of the system
and, lastly, introduce the averages:
\be
\langle . \rangle_a \equiv \easp_0 \easp_{1}...\easp_{a} ( f_{1} ... f_a \Omega_a(.) ).
\ee

We now introduce the generalized free energy $\tilde{\alpha}(t; x_1,...,x_K)$ mirroring the previous section.
\begin{definition}
The generalized free energy associated with the partition function $Z_0$ is defined as follows:
\be
\tilde{\alpha}(t; x_1,...,x_K) \equiv \frac{1}{N} \easp_0 \log Z_0 = \frac{1}{N} \easp_0 \easp_1 \log Z_1.
\ee
\end{definition}
We want to use this expression to write down a proper Hamilton-Jacobi equation, generalizing eq. (\ref{problema}) and
find the physical free energy in this enlarged space.
To this aim, we need the derivatives of the generalized free energy with respect to the interpolating parameters,
whose cumbersome computation is reported in the appendix.
\begin{lemma}
The derivatives of the generalized free energy with respect to the interpolating parameters are given by
\bea
\label{eq:dtalfa}
\partial_{t} \tilde{\alpha}_N(t; x_1,...,x_K) & = & \frac{1}{4} \left[1- \sum_{a=1}^K ( m_{a+1} - m_a) \langle q_{\sigma \sigma'}^p \rangle_a \right],\\
\label{eq:dxalfa}
\frac{\partial}{\partial x_a} \tilde{\alpha}_N(t; x_1,...,x_K) & \equiv & \partial_a \tilde{\alpha}_N (t; x_1,...,x_K) =  \frac{1}{2} \left( \frac{p}{2} Q_a^{p-2}(\beta) \right)^{1/2}
\left[ 1 - \sum_{b=a}^K (m_{b+1} - m_b) \langle q_{\sigma \sigma'} \rangle_b \right],
\eea
\end{lemma}
where we recall that
\be
\langle q_{\sigma \sigma'}^p \rangle_a  =  \easp_0 \easp_1 ... \easp_a (f_1...f_a \Omega_a(q_{\sigma \sigma'}^p) )
= \easp_0 \easp_1 ... \easp_a \left(f_1...f_a \frac{1}{N^p} \sum_{i_1,...,i_p} \omega_a^2(\sigma_{i_1}... \sigma_{i_p}) \right).
\ee
We are now ready to introduce the proper Hamilton principal function
in this generalized framework.
\begin{definition}
The Hamilton principal function is defined as follows
\be
S(t; x_1,...,x_K) \equiv 2 \tilde{\alpha} (t; x_1,...,x_K)
- \sum_{a=1}^K x_a \left(\frac{p}{2} Q_a^{ p-2}(\beta) \right)^{1/2}
- \frac{t}{2} \left[ 1 + \left(\frac{p}{2} - 1 \right) \sum_{a=1}^K (m_{a+1} - m_a) Q_a^p(\beta) \right].
\ee
\end{definition}
Using (\ref{eq:dtalfa}, \ref{eq:dxalfa}) we may easily compute its derivatives
\be
\label{eq:ders}
\begin{split}
\partial_t S(t; x_1,...,x_K) & =
- \frac{1}{2} \sum_{a=1}^K ( m_{a+1} - m_a) \langle q_{\sigma \sigma'}^p \rangle_a
-\left( \frac{p}{4} - \frac{1}{2} \right) \sum_{a=1}^K ( m_{a+1} - m_a) Q_a^p(\beta), \\
\partial_a S(t; x_1,...,x_K) & =
- \left( \frac{p}{2} Q_a^{p-2}(\beta) \right)^{1/2}  \sum_{b=a}^K (m_{b+1} - m_b) \langle q_{\sigma \sigma'} \rangle_b,
\end{split}
\ee
and write down the Hamilton-Jacobi equation which implicitly defines the potential $V(t; x_1,...,x_K)$
to whom our auxiliary mechanical system is subject:
\be
\label{eq:hj}
\partial_t S (t; x_1,...,x_K) + \frac{1}{2} \sum_{a, b = 1}^K \partial_a S(t; x_1,...,x_K) \times M^{-1}_{ab} \times  \partial_b S(t; x_1,...,x_K) + V(t; x_1,...,x_K) = 0.
\ee
Here $M^{-1}$ is the inverse of the mass matrix, which we are going to define in a convenient way through the kinetic energy $T(t; x_1,...,x_K)$:
\begin{definition}
The kinetic energy is defined as
\be
T(t; x_1,...,x_K)  \equiv
\frac{1}{2} \sum_{a, b = 1}^K \partial_a S(t; x_1,...,x_K) \times M^{-1}_{ab} \times \partial_b S(t; x_1,...,x_K).
\ee
\end{definition}
Using (\ref{eq:ders}), $T(t; x_1,...,x_K)$ may be written as
\be
\begin{split}
T(t; x_1,...,x_K) &
= \frac{p}{4} \sum_{a, b = 1}^K (M^{-1})_{ab} \left[Q_a(\beta) Q_b(\beta) \right]^{\frac{p-2}{2} }
\sum_{c\geq a}^K \sum_{d\geq b}^K (m_{c+1} - m_c) \langle q_{\sigma \sigma'} \rangle_c (m_{d+1} - m_d) \langle q_{\sigma \sigma'} \rangle_d \\
& = \frac{p}{4} \sum_{c, d = 1}^K D_{cd}
(m_{c+1} - m_c) \langle q_{\sigma \sigma'} \rangle_c (m_{d+1} - m_d) \langle q_{\sigma \sigma'} \rangle_d, \\
\end{split}
\ee
where we introduced the matrix $D$, whose generic entry is defined as
\be
D_{cd} \equiv \sum_{a=1}^c \sum_{b=1}^d (M^{-1})_{ab}  Q_a^{(p-2)/2}(\beta) Q_b^{(p-2)/2}(\beta).
\ee
To decouple the overlaps $\langle q_{\sigma \sigma'} \rangle_c $ and $\langle q_{\sigma \sigma'} \rangle_d$
we now pose
\be
\label{eq:dcd}
D_{cd} (m_{c+1} - m_c) = \delta_{cd} Q_c^{(p-2)/2}(\beta) Q_d^{(p-2)/2}(\beta),
\ee
where $\delta_{cd}$ is the Kronecker delta,
and then
\be
T(t; x_1,...,x_K) = \frac{p}{4} \sum_{a=1}^K (m_{a+1} - m_a) \langle q_{\sigma \sigma'} \rangle_a^2 Q_a^{p-2}(\beta).
\ee
\begin{definition} Within this mechanical analogy,  the potential $V(t; x_1,...,x_K)$ is, again, directly related to the fluctuations of the overlaps and can be introduced as follows:
\be
V(t; x_1,...,x_K) =
\frac{1}{2} \sum_{a=1}^K (m_{a+1} - m_a)
\{ \langle q_{\sigma \sigma'}^p \rangle_a - Q_a^p(\beta) + \frac{p}{2} \left[ Q_a^p(\beta) - \langle q_{\sigma \sigma'} \rangle_a^2 Q_a^{p-2}(\beta) \right] \}.
\ee
\end{definition}
The condition (\ref{eq:dcd}) completely determines the elements of $M^{-1}$.
These are all vanishing except on the diagonal and the terms whose indexes differ
only by one, which are symmetric:
\be
\begin{split}
(M^{-1})_{aa} & =
\frac{1}{m_{a+1} - m_a} + \frac{1}{m_a - m_{a-1}} \left(\frac{Q_a(\beta)}{Q_{a+1}(\beta)} \right)^{p-2}   \; \; \; a \geq 2, \\
(M^{-1})_{a, a+1}  & = (M^{-1})_{a, a+1}  =
-\frac{1}{m_{a+1} - m_a}  \left(\frac{Q_a(\beta)}{Q_{a+1}(\beta)} \right)^{(p-2)/2} \; \; \; a \geq 2,
\end{split}
\ee
and with
\be
(M^{-1})_{11} = \frac{1}{m_2}.
\ee
The matrix $M^{-1}$ clearly admits an inverse, its determinant being non-null:
\be
\det M^{-1} = \prod_{a=2}^{K+1} (m_a - m_{a-1}) \neq 0.
\ee
Notice that the elements of $M^{-1}$ and consequently $M$ depend on the overlaps $q_a$, differently from
the case $p=2$ \cite{tesispec}.
In this case, in fact, the system energy is no longer a quadratic form in the overlap averages,
and this has deep physical consequences; in particular,
the phase transition is first order for $p>2$, meaning
that the order parameter changes discontinuously at the critical temperature.

\subsection{The first step of broken replica symmetry}

Using results from the previous section, here we find out the expression of the free-energy corresponding to the first step of broken replica symmetry (1-RSB).
\begin{definition}
The generalized partition function and free-energy are defined as
\be
\begin{split}
\tilde{Z}_N (t; x_1, x_2) &
= \sum_{\sigma} \exp \left[ \sqrt{ \frac{tp!}{2N^{p-1}} } \sum_{1\leq i_1<...<i_p\leq N} J_{i_1,...i_p} \sigma_{i_1}...\sigma_{i_p}
+ \sum_{a=1}^2 \sqrt{x_a} \left( \frac{p}{2} Q_a^{p-2} \right)^{1/4} \sum_{i=1}^N J_i^a \sigma_i \right], \\
\tilde{\alpha}_N(t; x_1, x_2) &
= \frac{1}{N m} \easp_0 \easp_1 \log \easp_2 Z_2^m,
\end{split}
\ee
\end{definition}
where we took $m_2 \equiv m$ and we remind that in this case
\be
\begin{split}
Z_2 & \equiv \tilde{Z}_N,  \\
0 & = Q_0(\beta) < Q_1(\beta) < Q_2(\beta) < 1, \\
0 & = m_1 < m_2 < 1= m_3.
\end{split}
\ee
The principal Hamilton function $S(t,x_1, x_2)$ can be introduced as
\begin{definition}
The principal Hamilton function for the associated 1-RSB mechanical problem is
\be
\begin{split}
S(t,x_1, x_2)  = & \
2 \tilde{\alpha}(t,x_1, x_2)
- \left(\frac{p}{2} Q_1^{p-2}(\beta) \right)^{1/2} x_1 - \left(\frac{p}{2} Q_2^{p-2}(\beta) \right)^{1/2} x_2 \\
& - \frac{t}{2} \left[ 1 + \left( \frac{p}{2} -1 \right) (m Q_1^p(\beta) +(1-m) Q_2^p(\beta) ) \right].
\end{split}
\ee
\end{definition}
As shown in the general case in the previous sections, we must now evaluate its derivatives
\be
\begin{split}
\partial_t S(t,x_1, x_2) & =
- \frac{1}{2}  m \langle q_{\sigma \sigma'}^p \rangle_1 - \frac{1}{2}  ( 1 - m) \langle q_{\sigma \sigma'}^p \rangle_2
- \left(\frac{p}{4} - \frac{1}{2} \right) (mQ_1^p(\beta) + (1-m)Q_2^p(\beta)), \\
\partial_1 S(t,x_1, x_2) & =
- m \left(\frac{p}{2}Q_1^{p-2}(\beta) \right)^{1/2} \langle q_{\sigma \sigma'} \rangle_1
- (1-m)  \left( \frac{p}{2} Q_2^{p-2}(\beta) \right)^{1/2} \langle q_{\sigma \sigma'} \rangle_2, \\
\partial_2 S(t,x_1, x_2) & =
- (1-m)  \left( \frac{p}{2} Q_2^{p-2}(\beta) \right)^{1/2} \langle q_{\sigma \sigma'} \rangle_2.
\end{split}
\ee
\begin{proposition}
Choosing the inverse of the mass matrix (and so the mass matrix itself with the condition (\ref{eq:dcd}))
$$
M^{-1} =
\begin{bmatrix}
\frac{1}{m}  & - \frac{1}{m} (\frac{Q_1}{Q_{2}})^{(p-2)/2}  \\
- \frac{1}{m} (\frac{Q_1}{Q_{2}})^{(p-2)/2}    & \frac{1}{1 - m} + \frac{1}{m} (\frac{Q_1}{Q_2})^{p-2}
\end{bmatrix}
\;
\Rightarrow M =
\begin{bmatrix}
m + (1-m) (\frac{Q_1}{Q_{2}})^{p-2} &  (1-m) (\frac{Q_1}{Q_{2}})^{(p-2)/2}  \\
 (1-m) (\frac{Q_1}{Q_{2}})^{(p-2)/2}    & 1-m
\end{bmatrix}
$$
we can write down explicitly the kinetic term $T(t; x_1, x_2)$ and the potential $V(t; x_1, x_2)$ of the equivalent mechanical system:
\be
\begin{split}
T(t; x_1, x_2) & =
\frac{p}{4} m \langle q_{\sigma \sigma'} \rangle_1^2 Q_1^{p-2}(\beta)
 +\frac{p}{4} (1-m) \langle q_{\sigma \sigma'} \rangle_2^2 Q_2^{p-2}(\beta), \\
 V(t; x_1, x_2) & =
 \frac{1}{2} m \left[ \langle q_{\sigma \sigma'}^p \rangle_1 - Q_1^p(\beta)
 + \frac{p}{2}(Q_1^p(\beta) - \langle q_{\sigma \sigma'} \rangle_1^2 Q_1^{p-2}(\beta) )\right] \\
& +  \frac{1}{2} (1-m)
 \left[ \langle q_{\sigma \sigma'}^p \rangle_2 - Q_2^p(\beta)
 + \frac{p}{2}(Q_2^p(\beta) - \langle q_{\sigma \sigma'} \rangle_2^2 Q_2^{p-2}(\beta) )\right].
\end{split}
\ee
\end{proposition}
We can consequently state the following
\begin{proposition}
There is a mechanical analogy between the 1-RSB statistical mechanics of the p-spin-glass and an equivalent mechanical   system that moves in a two-dimensional space-time with equations of motion given by
\be
\label{eq:motion1rsb}
\begin{split}
x_1(t) & = x_1^0 + v_1(t; x_1, x_2) \, t, \\
x_2(t) & = x_2^0 + v_2(t; x_1, x_2) \, t. \\
\end{split}
\ee
The corresponding velocities are defined as
\be
\begin{split}
v_1(t; x_1, x_2) &
\equiv \sum_{a=1}^2 (M^{-1})_{1a} \partial_a S(t; x_1, x_2) =
 - \left( \frac{p}{2} Q_1^{p-2}(\beta) \right)^{1/2} \langle q_{\sigma \sigma'} \rangle_1 \\
v_2(t; x_1, x_2) &
\equiv \sum_{a=1}^2 (M^{-1})_{2 a} \partial_a S(t; x_1, x_2) =
\left( \frac{p}{2} \right)^{1/2}
 \frac{ Q_1^{p-2} (\beta)}{Q_2^{(p-2)/2}(\beta) } \langle q_{\sigma \sigma'} \rangle_1
 -\left( \frac{p}{2} \right)^{1/2} Q_2^{(p-2)/2}(\beta) \langle q_{\sigma \sigma'} \rangle_2.
\end{split}
\ee
\end{proposition}
As discussed before, we are interested in studying the free motion, i.e. the motion in absence of potential, and deduce the physical free-energy from the solution of the Hamilton-Jacobi equation
$$
\partial_t S(t;x_1,x_2) + \frac12 \sum_{a,b=1}^K \partial_a S (t;x_1,x_2) \times M^{-1}_{ab} \times \partial_b S (t;x_1,x_2).
$$
We stress that here the potential is related to a more complex kind of fluctuations of the overlap as we are requiring much more than the simple self-averaging: Physically we can think at each step of RSB as a refinement, a zoom, in the analysis of the free energy landscape, that allows to see rugged valleys otherwise averaged out and we are asking for adiabatic thermalization within each of these (sub)-valleys ("sub" w.r.t. the macro-ones already encoded in the RS-approximation).
Coherently, a sufficient condition for a vanishing 1-RSB potential is an overlap variance inside the bracket denoted with $\langle . \rangle_a$ equal to zero
and the identification of the averages of the overlap with the functions $Q_a(\beta)$:
\be
\label{eq:idq}
\langle q_{\sigma \sigma'}^2 \rangle_a = \langle q_{\sigma \sigma'} \rangle_a^2 = Q_a^2(\beta), \; a=1,2.
\ee
In the absence of a potential, the velocities (and so the kinetic energy) are conserved quantities and
we can then consider their values at the initial instant $t=0$, in perfect analogy with the RS case:
\be
\begin{split}
\bar{q_1}  \equiv \langle q_{\sigma \sigma'} \rangle_1(0; x_1^0, x_2^0) &
= \int d\mu (z_1) \bigg[
\frac{ \int d\mu (z_2) \cosh^m \theta (z_1, z_2) \tanh \theta (z_1, z_2)}{ \int d\mu (z_2) \cosh^m \theta (z_1, z_2) }
\bigg]^2, \\
\bar{q_2}  \equiv \langle q_{\sigma \sigma'} \rangle_2(0; x_1^0, x_2^0) &
= \int d\mu (z_1)
\frac{ \int d\mu (z_2) \cosh^m \theta (z_1, z_2) \tanh^2 \theta (z_1, z_2)}{ \int d\mu (z_2) \cosh^m \theta (z_1, z_2) }, \\
\theta (z_1, z_2) &
\equiv \sqrt{x_1^0} \left( \frac{p}{2} Q_1^{p-2} \right)^{1/4} z_1 + \sqrt{x_2^0} \left( \frac{p}{2} Q_2^{p-2} \right)^{1/4} z_2,
\end{split}
\ee
where $\theta$ will be defined in eq. (\ref{eq:teta}) and 
\be
d\mu(z) = \exp(-z^2/2) dz
\ee
is the Gaussian measure.
This computation essentially leads us to the 1-RSB self-consistence equations for overlaps
when considering the statistical-physics point $t=\beta^2, x_1 = x_2 = 0$.
In this point, and with the condition (\ref{eq:idq}), the equations of motion (\ref{eq:motion1rsb}) give
\be
\begin{split}
x_1^0 &
= \beta^2 \left( \frac{p}{2} Q_1^{p}(\beta) \right)^{1/2}, \\
x_2^0 &
=
 \beta^2 \left( \frac{p}{2} \right)^{1/2} Q_2^{\frac{p}{2} }(\beta)
 -\beta^2 \left( \frac{p}{2} \right)^{1/2}
 \frac{ Q_1^{p-1} (\beta)}{Q_2^{\frac{p-2}{2} }(\beta) },
\end{split}
\ee
so that the explicit self-consistence equations contain
\be
\label{eq:teta}
\theta (z_1, z_2)
\equiv \beta \left(\frac{p}{2}\right)^{1/2} z_1 Q_1^{\frac{p-1}{2} }(\beta)
+ \beta \left(\frac{p}{2}\right)^{1/2} z_2 \sqrt{Q_2^{p-1}(\beta)  - Q_1^{p-1}(\beta) }.
\ee
\begin{remark}
In the second term of the r.h.s. of equation (\ref{eq:teta}) the two overlaps are decoupled, and in the limit $p \to 2 $ we get
the correct 1-RSB self-consistence equation for the SK model too.
\end{remark}
To compute the free-energy we use the usual recipe: As we assume that the mechanical potential is
zero, we write (easily) the solution for the Hamilton-Jacobi problem and then we evaluate it in the point $t=\beta^2, x_a=0$.
First of all, we need the free-energy at the initial instant, which is straightforward to obtain, since it
contains no spin interactions:
\be
\tilde{\alpha}(0; x_1^0, x_2^0) = \log 2 + \frac{1}{m} \int d\mu (z_1) \log \int d\mu(z_2)
\cosh^m \theta(z_1, z_2),
\ee
with $\theta(z_1, z_2)$ given by (\ref{eq:teta}).
The Hamilton function which is solution of (\ref{eq:hj}) for a vanishing potential $V\equiv 0$
is simply given by the function at the initial instant plus the integral of the Lagrangian,
(which corresponds to the kinetic energy only), over time
\be
S(t; x_1, x_2) =
S(0; x_1^0, x_2^0) + \int_0^t ds T(s; x_1, x_2) =
S(0; x_1^0, x_2^0) +  T(0; x_1^0, x_2^0) t,
\ee
where we used the fact that the kinetic energy is a conserved quantity.
We obtain in this way
\be
\begin{split}
S(t; x_1, x_2) =  &
 2 \tilde{\alpha}( 0; x_1^0, x_2^0) - \left(\frac{p}{2} Q_1^{p-2}(\beta) \right)^{1/2} x_1^0 - \left(\frac{p}{2} Q_2^{p-2}(\beta) \right)^{1/4} x_2^0  \\
&  + \frac{t p}{4} Q_1^{p}(\beta)
 +\frac{t p}{4} (1-m)  Q_2^{p}(\beta),
\end{split}
\ee
and from this, the generalized free-energy $\tilde{\alpha}(t; x_1, x_2)$ 
\be
\begin{split}
\tilde{\alpha}(t; x_1, x_2) = & \frac{1}{2} S(t; x_1, x_2)
+ \frac{1}{2}  \left(\frac{p}{2} Q_1^{p-2}(\beta) \right)^{1/2} x_1 + \frac{1}{2}\left(\frac{p}{2} Q_2^{p-2}(\beta) \right)^{1/2} x_2 \big)\\
& + \frac{t}{4} \left[ 1 + \left(\frac{p}{2} - 1 \right) ( m Q_1^p(\beta) + (1-m) Q_2^p(\beta) )
\right].
\end{split}
\ee
Then the physical free-energy is easily computed by taking $t=\beta^2, x_1=x_2=0$ and we can state the next theorem:
\begin{theorem}
Making the assumption of vanishing potential $V(t,x_1,x_2)$ in the mechanical analogy, the corresponding free energy for the p-spin glass model corresponds to the so called ``1-RSB'' and is given by
\be
\begin{split}
\alpha(\beta) = &
\log 2 + \frac{1}{m} \int d\mu (z_1) \log \int d\mu(z_2)
\cosh^m \big( \beta z_1 Q_1^{(p-1)/2}(\beta)  + \beta z_2 \sqrt{Q_2^{p-1}(\beta) - Q_1^{p-1}(\beta) } \big) \\
& + \frac{\beta^2}{4} \left[ 1 + (p-1)mQ_1^p(\beta) + (p-1)(1-m)Q_2^p(\beta) - p Q_1^{p-1}(\beta) - p Q_2^{p-1}(\beta) \right].
\end{split}
\ee
\end{theorem}
We skip here any digression on the physics behind these formulas as these are in perfect agreement with the original investigation by Gardner \cite{gardner} and by Gross and Mezard \cite{mezardgross}, so to highlight only the mathematical methods, to which this paper is dedicated.

\subsection{Conservation laws: Polynomial identities}

We conclude this section with an analysis of the conserved quantities deriving from the internal symmetries of the theory. We will approach them as N\"{o}ther integrals within the Hamilton-Jacobi formalism, while at the end of the next section we will re-obtain (and discuss more deeply) the same constraints within a more familiar thermodynamic approach.

Let us restate the Hamilton-Jacobi equation
$$
\partial_t S(t,x)+H(\partial_x S(t,x),t,x)=0
$$
where the Hamiltonian function reads off as \cite{note1}
\begin{equation}
H(\partial_x S(t,x),t,x)= T(t,x) +V(t,x).
\end{equation}
Hamilton equations are nothing but characteristics given by:
\begin{equation}\label{eq:hamilton}
\left\{
\begin{array}{rcl}
 \dot{x}&=&v(t,x)\\
 \dot{t}&=&1\\
 \dot{P}&=&-v(t,x)\partial_x v(t,x)-\partial_x V(t,x)\\
 \dot{E}&=&-v(t,x)\partial_x\left(\partial_t S(t,x)\right)-\partial_t
 V(t,x),
\end{array}\right.
\end{equation}
the latter two equations display space-time translational invariance and express the conservation laws for
momentum and energy for our system, further, these can be written in form of
streaming equations as
\begin{equation*}
\left\{
\begin{array}{rcl}
D P(t,x)&=&-\partial_x V(t,x)\\
D \partial_t S(t,x)&=&-\partial_t V(t,x).
 \end{array}
\right.
\end{equation*}
Since we are interested in evaluating the free motion,
bearing in mind that $v(x,t)=-\langle q_{12}^{p/2} \rangle$ and
$\partial_t S(x,t)=-\frac{1}{2}\langle q_{12}^p \rangle$, so
$D=\partial_t-\langle q^{p/2}\rangle \partial_x$, we
conclude
\begin{equation}\label{eq:cons-ising}
\left\{
\begin{array}{rcl}
D \langle  q^{p/2} \rangle &=& 0,\\
D \langle q^p \rangle &=& 0,
 \end{array}
\right.
\end{equation}
\textit{i.e.}
\begin{equation}\label{eq:cons-ising}
\left\{
\begin{array}{rcl}
\langle q_{12}^{p} \rangle - 4 \langle q_{12}^{p/2}q_{13}^{p/2} \rangle + 3 \langle q_{12}^{p/2}q_{34}^{p/2} \rangle &=&  0,\\
\langle q_{12}^{2p} \rangle - 4 \langle q_{12}^{p}q_{23}^{p} \rangle + 3 \langle q_{12}^{p}q_{34}^{p} \rangle &=&  0.
\end{array}
\right.
\end{equation}
\begin{remark}
The orbits of the N\"other groups of the theory coincide with the streaming lines of the Hamilton-Jacobi Hamiltonian, and conservation laws along these lines give well known identities in
the statistical mechanics of the model often known as Ghirlanda-Guerra relations and Aizenman-Contucci identities \cite{gg,aizcont}.
\end{remark}
We will deserve Sections VC and VD of the paper to deepen our understanding of these identities within the smooth cavity field approach, hence we do not investigate them further here.

\section{Second approach: The smooth cavity field.}

\subsection{Smooth cavity field and stochastic stability}

The main heuristic idea of the {\em cavity field} method is to look for an
explicit expression of $\alpha(\beta)=-\beta f(\beta)$ upon
increasing the size of the system from $N$ particles 
to $N+1$ (originally the technique was developed by removing a spin instead of adding, hence "cavity", but we will follow the approach recently developed in \cite{irriducibile}). As a consequence, within this framework attention will be payed at the system size and all the $N$ dependencies will be explicitly introduced.
\newline
On the other hand, in order to formulate stochastic stability, we
have to consider the statistical properties of the system with a
Hamiltonian given by the original Hamiltonian $H$ plus a random
perturbation $\tilde{H}$ so to write $H' = H + \epsilon \tilde{H}$.
Stochastic stability states that all the properties of the system
are smooth functions of $\epsilon$ around $\epsilon = 0$, after
the appropriate averages over the original Hamiltonian and the
random Hamiltonian have been taken. We stress that, even though
initially it was only postulated \cite{aizcont}, stochastic stability has
recently \cite{contuccistocstab} been rigorously proven for a wide class
of disordered Hamiltonians.

Our idea, to be explored in detail later on, is that for a
system with a gauge-invariant Hamiltonian (like the even p-spin model at
zero external field) we can choose, as generic random
perturbation $\tilde{H}$ in the stochastic stability approach, a
term  proportional to $ \sum_{i_1 < ... < i_{p-1}} J_{i_1,...,i_{p-1}} \sigma_{i_1}...\sigma_{i_{p-1}}$.
Here the $J_{i_1,...,i_{p-1}}$ are random
fields, taken from the same Gaussian i.i.d.\ distribution as the
original $J_{i_1,...,i_p}$. The key insight is that this is a ``hidden''
cavity field: by applying the transformation $\sigma_i \rightarrow
\sigma_i\sigma_{N+1}$ $\forall i$ (which leaves the Hamiltonian $H$ 
invariant) it is possible to switch the stochastic stability
approach into the standard cavity field approach. As we are going
to see, this technique offers more freedom than the two single,
 non interacting, approaches as we can turn one into the other as desired.

To explain the method we need some preliminary definitions. First, let
us introduce an extended partition function that
includes an interaction with an added hidden spin $\sigma_{N+1}$
through a control parameter $t \in [0,\beta^2]$ such that for
$t=0$ we have the classical partition function of $N$ spins while
for $t=\beta^2$ we get the partition function
 (times one half) for the larger system, with a little temperature shift which
vanishes in the thermodynamic limit:
\be\label{cf10} Z_{N}(\beta,t)= \sum_{\sigma}e^{-\beta H_N(\sigma;J)
+\sqrt{\frac{t p!}{2N^{p-1}}}\sum_{1 \leq i_1 < ... < i_{p-1} \leq N}J_{i_1,...,i_{p-1}}\sigma_{i_1}...\sigma_{i_{p-1}}}. \ee
Indeed, when $t = \beta^2$,  by redefining $J_i \rightarrow
J_{i,N+1}$ and making the transformation $\sigma_i \rightarrow
\sigma_i\sigma_{N+1} \ \forall i$, we obtain the partition function
for a system of $N+1$ spins at a shifted temperature $\beta^*$ such
that \be\label{cfx} \beta^*=\beta \Big((N+1)/N\Big)^{\frac{p-1}{2}} \rightarrow \beta
\mbox{ for } N \rightarrow \infty. \ee The only other, trivial,
difference is that of course the sum over $\sigma_{N+1}$ in the
partition function for $N+1$ spins gives an additional factor $2$.

Next, we state the two key symmetries whose breaking we will be
concerned with. These apply to the unperturbed ($t=0$) system;
recall that $\langle.\rangle\equiv \mathbb{E}\omega(.)$.
\begin{proposition}{\itshape
The averages $\langle \cdot \rangle$ are replica-symmetric, i.e.
invariant under permutation of replicas. In other words, for any
function $F_s(\{q_{ab}\})$ of the overlaps among $s$ replicas and
any permutation $g$ of $s$ elements, $\langle F_s(\{q_{ab}\})
\rangle = \langle F_s(\{q_{g(a)g(b)}\}) \rangle$}.
\end{proposition}
Note that there is no issue with replica symmetry breaking here,
as we are concerned with {\em real} replicas.
\bigskip

\begin{proposition}
\label{gauge_symmetry}{\itshape The averages $\langle \cdot \rangle$
are invariant under gauge transformation, i.e.\ for any assignment
of the $\epsilon^a=\pm1$ we have}
\be \langle F_s(\{q_{ab}\})\rangle = \langle F_s(\{\epsilon^a
\epsilon^b q_{ab}\})\rangle\ . \ee
\end{proposition}
This second symmetry is a consequence of the fact that the
Hamiltonian (in zero field) is even in the spins, i.e.\ it remains
unchanged when we transform $\sigma_i^a\to \epsilon^a \sigma_i^a$.
\newline
Next we will formalize some terminology and concepts which will be
useful for developing our {\sl smooth} version of the cavity
method.
\begin{definition}{\itshape
We define as ``filled'' a monomial of the overlaps in which every
replica appears an even number of times.}
\end{definition}
\begin{definition}{\itshape
We define as ``fillable'' a monomial of the overlaps in which the
above property is obtainable by multiplying with exactly one
two-replica overlap.}
\end{definition}
\begin{definition}{\itshape
We define as ``unfillable'' a monomial which is neither filled nor
fillable.}
\end{definition}

\bigskip

Polynomials that are sums of filled monomials will themselves be
called filled, etc. We give a few examples:
\begin{itemize}
\item The monomials $q_{12}^p$ and $q_{12}^p q_{34}^p$ are filled (as $p$ is even by definition).

\item The monomial $q_{12}^{p-1}$ is fillable: multiplication by $q_{12}$
  gives the filled monomial $q_{12}^p$. Similarly $q_{12}^p q_{34}^{p-1}$ is
  fillable: it is filled by multiplication with $q_{34}$.

\item The following monomials are unfillable: $q_{12}^{p-1}q_{34}^{p-1}$, \
  $q_{12}q_{23}q_{45}$.
\end{itemize}
Now the plan to gain information on the p-spin-glass free energy as follows: First we define the cavity function and we prove some properties (related stochastic stability) of the classes of overlap monomials defined above. Then, we show that the free energy can be written as the internal energy plus the cavity function, and lastly, we expand the cavity function through the overlap monomials. Merging all together we have an irreducible expression of the free energy in terms of overlap monomials (which physically correspond to overlap correlation functions).
\begin{definition}{\itshape
We define the {\em cavity function}  $\Psi_N(\beta,t)$ as:} \be
\label{cf12}
\Psi_N(\beta,t)=\mathbb{E}[\ln\omega(e^{\sqrt{\frac{t p!}{2N^{p-1}}}\sum_{1 \leq i_1 < ... < i_{p-1} \leq N}J_{i_1,...,i_{p-1}}\sigma_{i_1}...\sigma_{i_{p-1}}})]
=\mathbb{E}\left[\ln\frac{Z_{N}(\beta,t)}{Z_N(\beta)}\right]. \ee
\end{definition}
\begin{definition}{\itshape
We define the generalized Boltzmann state that corresponds
to the partition function (\ref{cf10}) as:} \be \label{cf0}
\omega_t(F)=\frac{\omega(Fe^{ \sqrt{\frac{t p!}{2N^{p-1}}}\sum_{1 \leq i_1 < ... < i_{p-1} \leq N}J_{i_1,...,i_{p-1}}\sigma_{i_1}...\sigma_{i_{p-1}}})}{\omega(e^{\sqrt{\frac{t p!}{2N^{p-1}}}
\sum_{1 \leq i_1 < ... <  i_{p-1} \leq N}J_{i_1,...,i_{p-1}}\sigma_{i_1}...\sigma_{i_{p-1}}})}, \ee {\itshape where $F$ is a generic function of
the $N$-spin  configuration $\sigma$.}
\end{definition}
The next step  is to motivate why we have introduced these
definitions. We will first state two Theorems (\ref{richard} and
\ref{benson}) that show that the filled and the fillable
monomials have peculiar properties. Monomials in the first class
do not depend on the perturbation (i.e.\ they are stochastically
stable) while those in the second class become filled (via the
$\sigma_i \rightarrow \sigma_i\sigma_{N+1}$ gauge transformation)
in the thermodynamic limit.

\begin{theorem}\label{richard}
{\itshape In the $N \rightarrow \infty$ limit  the averages
$\langle Q \rangle$ of the filled monomial $Q$ are $t$-independent
for almost all values of $\beta$, such that}
\end{theorem}
$$
\lim_{N \rightarrow \infty} \partial_t \langle Q \rangle_t = 0
$$
\begin{proof}
We will prove the theorem in a key case, namely for $Q= q^p_{12}$, and refer to \cite{irriducibile} for further generalizations.
Let us write the cavity function as \be\label{cf13}
\Psi_N(\beta,t)=\mathbb{E}[\ln Z_{N}(\beta,t)]-\mathbb{E}[\ln
Z_{N}(\beta)], \ee
and take its derivative with respect to $\beta$ (writing again
$\langle.\rangle_t\equiv \mathbb{E}\omega_t(.)$), we have: \be \label{cf16}
\partial_\beta\Psi_N(\beta,t)=\frac {\beta N}{2}( \langle
q_{12}^p \rangle - \langle q_{12}^p\rangle_t). \ee We want to show
now that the function $\Upsilon_N(\beta,t)= \langle q_{12}^p \rangle -
\langle q_{12}^p \rangle_t$ vanishes for $N\to\infty$. From eq. (\ref{cf16}) we have
\be
\Upsilon_N(\beta,t)=\frac{4}{N}\partial_{\beta^2} \Psi_N(\beta,t).
\ee
and integrating this in a generic interval $[\beta_1^2,\beta_2^2]$
gives \be\label{contesto}
\int_{\beta_1^2}^{\beta_2^2}\Upsilon_N(\beta,t)d\beta^2=
\frac{4}{N}[\Psi_N(\beta_2,t)-\Psi_N(\beta_1,t)]. \ee
To finish the proof, we show that $\Psi_N(\beta,t)$ is of order
unity. The simplest way to do this is by looking at its
``streaming'', i.e.\ its variation with $t$. By a direct calculation
  one finds
\be\label{deripsi}
\frac{d \Psi_N(\beta,t)}{dt}
=\frac{1}{2}\mathbb{E}\left[1-\frac{1}{N}\sum_{1 \leq i_1 < ... < i_{p-1} \leq
N}\omega_t^2(\sigma_{i_1}...\sigma_{i_{p-1}})\right]= \frac{1}{2}( 1 - \langle q_{12}^{p-1}
\rangle_t). \ee
Hence, since $\langle q_{12}^{p-1}\rangle_t \in [-1,1]$, and
with $\Psi_N(\beta,0)=0$ (due to
$Z_{N}(\beta,t=0)=Z_N(\beta)$), we have $0\leq \Psi_N(\beta,t)\leq t$.
Therefore the r.h.s.\ of (\ref{contesto}) goes to zero for
$N\to\infty$, and the same holds for the average of
$\Upsilon_N(\beta,t)$ over any small temperature interval (with
the exception of singularities).
Consequently, $\Upsilon_N(\beta,t)$ itself goes to zero,
implying the claimed $t$-independence of the filled overlap monomials $\langle q_{12}^p \rangle_t
\rightarrow \langle q_{12}^p \rangle$.
\end{proof}
The next Theorem is crucial for this section, so we first prove a
lemma which contains the core idea. We temporarily introduce
subscripts on the Boltzmann states to clearly distinguish the
different quantities considered.

\begin{lemma}\label{Rindex}  Let  $\omega_{N,\beta}(.)$ and
$\omega_{N,\beta,t}(.)$ be the Boltzmann states defined, on a
system of $N$ spins, respectively by the canonical partition
function and by the extended one (\ref{cf10}). Consider a set of
$r$ distinct spin sites $\{i_1,..,i_r\}$ with $1\leq r\leq N$.
Then for $t=\beta^2$, the extended state becomes comparable to the
canonical state of an $N+1$ spin system, in that the following
relation holds
\be\label{bah} \omega_{N,\beta,t=\beta^2}(\sigma_{i_1}\cdots\sigma_{i_r})=
\omega_{N+1,\beta^*}(\sigma_{i_1}\cdots\sigma_{i_r}\sigma_{N+1}^r)
. \ee Note that the $r$ in the last factor is an exponent, not a
replica index, so that $\sigma_{N+1}^r=1$ if $r$ is even and
$\sigma_{N+1}^r=\sigma_{N+1}$ if $r$ is odd.
\end{lemma}
\begin{proof}
The proof is based on an application of the gauge symmetry,
i.e.\ the substitution $\sigma_i \rightarrow
\sigma_i\sigma_{N+1}$. Let us write out explicitly the l.h.s.\ of
eq. (\ref{bah}), abbreviating $\pi\equiv
\sigma_{i_1}\cdots\sigma_{i_r}$:
\be \omega_{N,\beta,t=\beta^2}(\pi)=\mathbb{E}
\frac{\sum_{{\sigma}}\pi e^{-\beta H_N(\sigma,J)+ \beta \sqrt{\frac{ p!}{2N^{p-1}}}\sum_{1 \leq i_1 < ... < i_{p-1} \leq N}J_{i_1,...,i_{p-1}}\sigma_{i_1}...\sigma_{i_{p-1}}}}{\sum_{{\sigma}}
e^{-\beta H_N(\sigma,J) + \beta \sqrt{\frac{ p! }{2N^{p-1}}}\sum_{1 \leq i_1 < ... < i_{p-1} \leq N}J_{i_1,...,i_{p-1}}\sigma_{i_1}...\sigma_{i_{p-1}}}} . \ee
Introducing a sum over $\sigma_{N+1}$ into the numerator and the
denominator (which is the same as multiplying and dividing by $2$
because there is no dependence on $\sigma_{N+1}$) and making the
transformation  $\sigma_i\rightarrow \sigma_i\sigma_{N+1}$, the
factor $\pi$ in the numerator is transformed into $\pi
\sigma_{N+1}^r$. The exponential becomes the extended Boltzmann
factor of an $(N+1)$-spin system at the modified temperature
(\ref{cfx}), so that
\be \omega_{N,\beta,t= \beta^2}(\pi)=
\omega_{N+1,\beta^*}(\pi\sigma_{N+1}^r)\ee
as claimed.
\end{proof}

Using this lemma, it is straightforward to prove the following theorem, whose proof we omit as it is identical to the one shown in \cite{irriducibile}.
\begin{theorem}\label{benson}
{\itshape Let $Q$ be a fillable overlap monomial, such that
$q_{ab}Q$ is filled. Then for $N\to\infty$}
\be \langle Q \rangle_{t=\beta^2} = \langle q_{ab} Q \rangle,
\ee {\itshape where the average on the right is evaluated in the
canonical  Boltzmann state} ($t=0$). {\itshape We will refer to
this property as saturability.}
\end{theorem}

To motivate physically why theorem \ref{richard} should indeed be true for all
filled monomials, let us make a clear example: Suppose that such a monomial $Q$ is a function of
overlaps among $s$ replicas. Consider as before the Boltzmann
measure perturbed by a smooth cavity field and call $\sigma^{a}$
the $N$-spin configuration of the replica $a$. We apply the gauge
transformation $\sigma_i^a \rightarrow \sigma_i^a\sigma_{N+1}^a$,
calling $\sigma_+^a=(\sigma_1^a,\ldots,\sigma^a_{N+1})$ the
enlarged spin vector obtained. The key feature of a filled
monomial $Q$ is that it is left invariant by this transformation,
so that (all sums run over $a=1\ldots s$ and $i=1\ldots N$)
\begin{eqnarray} \nonumber
\langle Q \rangle_t &=&
\mathbb{E}\frac{\sum_{\{\sigma_{N+1}^a\}}\sum_{\{\sigma^a\}}Q(\{q_{ab}^{(N)}\})
e^{-\sum_{a} \beta H(\sigma^a)+ \sqrt{\frac{t p!}{2N^{p-1}}}\sum_{i_1<...< i_{p-1}} J_{i_1...i_{p-1}} \sigma_{i_1}^a ... \sigma_{i_{p-1}}^a}}
{\sum_{\{\sigma_{N+1}^a\}}\sum_{\{\sigma^a\}}
e^{-\sum_{a} \beta H(\sigma^a)+ \sqrt{\frac{t p!}{2N^{p-1}}}\sum_{i_1<...< i_{p-1}} J_{i_1...i_{p-1}} \sigma_{i_1}^a ... \sigma_{i_{p-1}}^a}}
\\ \nonumber &=& \mathbb{E}\frac{\sum_{\{\sigma^a_+\}}Q(\{q_{ab}^{(N)}\})
e^{-\sum_{a}\beta H(\sigma^a)+ \sqrt{\frac{t p!}{2N^{p-1}}}\sum_{i_1<...< i_{p-1}} J_{i_1...i_{p-1}} \sigma_{i_1}^a ... \sigma_{i_{p-1}}^a\sigma_{N+1}^a}}
{\sum_{\{\sigma^a_+\}}e^{-\sum_{a}\beta
H(\sigma^a)+ \sqrt{\frac{t p!}{2N^{p-1}}}\sum_{i_1<...< i_{p-1}} J_{i_1...i_{p-1}} \sigma_{i_1}^a ... \sigma_{i_{p-1}}^a\sigma_{N+1}^a}} \\
\nonumber &=&
\mathbb{E}\frac{\sum_{\{\sigma^a_+\}}Q(\{q_{ab}^{(N+1)}+O(N^{-1})\})
e^{-\sum_{a}\beta H(\sigma^a)+\sqrt{t/N}\sum_{i,a}
J_i\sigma_i^a\sigma_{N+1}^a}} {\sum_{\{\sigma^a_+\}}e^{-\sum_{a}
\beta
H(\sigma^a)+\sqrt{t/N}\sum_{i,a} J_i\sigma_i^a\sigma_{N+1}^a}} \\
&=&
\mathbb{E}\frac{\sum_{\{\sigma^a_+\}}Q(\{q_{ab}^{(N+1)}\})e^{-\sum_{a}
\beta^* H_+(\sigma_+^a)}} {\sum_{\{\sigma^a_+\}}e^{-\sum_{a}
\beta^* H_+(\sigma_+^a)}} + O\left(\frac{1}{N}\right),
\label{t_independence_motivation}
\end{eqnarray}
with $\beta^*$ defined as before and
\begin{eqnarray}\nonumber
H_+(\sigma_+^a) = &-& \sqrt{\frac{t p!}{2 (N+1)^{p-1}}}\sum_{1\leq i_1 < ... < i_{p} \leq
N}J_{i_1,...,i_p}\sigma_{i_1}...\sigma_{i_p} \\ &-&\sqrt{\frac{t p!}{2 (N+1)^{p-1}}} \sum_{1\leq i_1 < ... < i_{p-1} \leq
N} \sqrt{\frac{t}{\beta^2}}J_{i_1...i_{p-1}}\sigma_{i_1}...\sigma_{i_{p-1}}\sigma_{N+1}.
\end{eqnarray}
For $t=\beta^2$ we have an $N+1$ spin system at the slightly
shifted temperature $\beta^*$, and for $N\to\infty$ this will give
the same result as for an $N$ spin system at the original
temperature up to vanishingly small corrections: $\langle
Q\rangle_{t=\beta^2}=\langle Q\rangle+O(1/N)$. For generic nonzero
$t$ one has in addition a modified strength of the interaction of
one spin ($\sigma_{N+1}$) with all others. Also, this should only
give $O(1/N)$ corrections because to produce a non-vanishing perturbation
one expects that a finite fraction of spins should have
non-standard interaction strengths.
\newline
\newline
As a final ingredient for later developments, let us show the
streaming of a generic observable, i.e.\ its variation w.r.t. the
parameter $t$ ruling the strength of the smooth cavity
perturbation, that we state without the proof as it is a long but straightforward generalization of the once given for instance in \cite{irriducibile,boh}:
\begin{proposition}
\label{streaming_theorem} {\itshape Let $F_s$  be a monomial of
overlaps among $s$ replicas;
  then for any $N$ the following streaming equation for $F_s$ holds:}
\be\label{108}
\frac{\langle F_s \rangle_t}{dt} = \langle F_s (\sum_{1\leq a < b \leq
  s}q^{p-1}_{ab} -s\sum_{1 \leq  a \leq s}q^{p-1}_{a,s+1}+ \frac{s(s+1)}{2}q^{p-1}_{s+1,s+2})
\rangle_t. \ee
\end{proposition}

\subsection{Stochastically stable expansions}

For the sake of clearness, let us outline briefly the plan for this section: first we
link the free energy, the internal energy and the cavity function
(which carries the information about the entropy). As we are
interested in the free energy and an explicit expression for
the internal energy is obtained through a direct calculation as 
$$
\langle H_N(\sigma, J) \rangle = -\frac{\beta}{2}\left(1- \langle q_{12}^p \rangle \right),
$$
our attention is focused on the cavity function: We show that it is possible to represent it in terms of
filled overlap monomials. These are evaluated initially in the
perturbed Boltzmann state $\omega_t$ but because they are
stochastically stable according to theorem \ref{richard}, we can also
evaluate them in the unperturbed state. Adding the internal energy part then gives us the
desired expansion of the free energy in terms of overlap
correlation functions as stated by the following theorem.
\begin{theorem}\label{baffo}
{\itshape Assuming that the infinite volume limit of the cavity
function} $$\Psi(\beta, t= \beta^2) =  \lim_{N \to \infty}\Psi_N(\beta, t= \beta^2)$$ {\itshape is well behaved,
the following relation holds in the thermodynamic limit}:
\be\label{215} \alpha(\beta)+\frac {\beta}{2}(p-1)\partial_\beta
\alpha(\beta)
 = \ln 2+\Psi(\beta,t=\beta^2). \ee
\end{theorem}
\begin{proof}
Let us consider
the partition function of a system of $N+1$ spins and at an
inverse temperature $\beta^*$, which is slightly larger than the
``true'' inverse temperature $\beta$ according to (\ref{cfx}).
Then, using the gauge transformation $\sigma_i \to
\sigma_i \sigma_{N+1}$ in reverse, we get
\begin{eqnarray}
Z_{N+1}(\beta^*) &=&\sum_{\{\sigma\},\sigma_{N+1}}
e^{\frac{\beta^* \sqrt{p!}}{\sqrt{2 (N+1)^{p-1}}}\sum_{1 \leq i_1  < ... < i_p \leq
N+1}J_{i_1,...,i_p}\sigma_{i_1}...\sigma_{i_p}}
\nonumber\\
&\sim& \nonumber  2\sum_{\sigma} e^{\frac{\beta \sqrt{p!}}{\sqrt{2 N^{p-1}}}\sum_{1 \leq i_1  < ... < i_p \leq
N}J_{i_1,...,i_p}\sigma_{i_1}...\sigma_{i_p}}e^{\frac{\beta \sqrt{p!}}{\sqrt{2 N^{p-1}}}\sum_{1 \leq i_1  < ... < i_{p-1} \leq N}J_{i_1,...,i_{p-1}}\sigma_{i_1}...\sigma_{i_{p-1}}}\\
&=& \label{217}  2 Z_N(\beta) \omega_{N,\beta}(
e^{\frac{\beta \sqrt{p!}}{\sqrt{2 N^{p-1}}}\sum_{1 \leq i_1  < ... < i_{p-1} \leq
N}J_{i_1,...,i_{p-1}}\sigma_{i_1}...\sigma_{i_{p-1}}})
. \end{eqnarray}
Taking logarithms and averaging over the disorder, the last term
just becomes the cavity function (as the $J_{i,N+1}$ have the same
distribution as the $J_i$ in the original definition):
\begin{eqnarray*}
&&[\mathbb{E} \ln Z_{N+1}(\beta^*) - \mathbb{E} \ln
Z_{N+1}(\beta)] + [\mathbb{E} \ln Z_{N+1}(\beta) - \mathbb{E} \ln
Z_{N}(\beta)]
\\
&&\qquad\qquad\quad = \ln 2 + \Psi_N(\beta,t=\beta^2) .
\end{eqnarray*}
The first combination in square brackets on the l.h.s.\ can now be
expanded in the small difference
\be \beta^*-\beta=\beta\left[ \left(\frac{N+1}{N}\right)^{\frac{(p-1)}{2}}-1\right]= (p-1)\frac{\beta}{2N} +
O(\frac{1}{N^2}), \ee
according to
\begin{eqnarray*} \mathbb{E} \ln
Z_{N+1}(\beta^*)- \mathbb{E}\ln Z_{N+1}(\beta) &=&
(p-1)\frac{\beta}{2N}
\partial_{\beta} \mathbb{E}\ln Z_{N+1}(\beta)+ O(1/N)
\\&=&
(p-1)\frac{\beta}{2} \partial_\beta \alpha_{N+1}(\beta)+O(1/N).
\end{eqnarray*}
The difference in the second set of square brackets will give the
pressure $\alpha(\beta)$ for large $N$, and taking $N\to\infty$
therefore directly gives the statement of the theorem.
%
\end{proof}
Strictly speaking, the existence of
the thermodynamic limit is not sufficient to guarantee that the
free energy increments converge, as assumed above. This technical
difficulty can be avoided by taking a Ces\`aro limit (see for instance \cite{barradesanctis})
rather than a standard limit $N\to\infty$, and the
large-$N$ value of the cavity function then should be understood
in this sense.
\newline
This theorem states that we
need to study the cavity function to extrapolate properties
 of the free energy. To do this, let us recall its
streaming w.r.t. $t$, as given in (\ref{deripsi}): \be
\frac{d \Psi_N(\beta,t)}{dt} =
\frac{p}{4}(1- \langle q^{p-1}_{12} \rangle_t). \label{deripsi_again}
\ee
Since the cavity function vanishes for $t=0$, it can then be
written as 
\be\label{666} \Psi_N(\beta,t) = \frac{p}{4}\int_0^t dt'\,(1-
\langle q_{12}^{p-1}\rangle_{t'}). \ee
The plan now is to expand $\langle q^{p-1}_{12} \rangle_t$ in $t$, by evaluating successive
$t$-derivatives via the streaming equation (proposition 
\ref{streaming_theorem}). A key insight that makes this expansion
possible is that at $t=0$ all averages of monomials that are not
filled must vanish because they would otherwise acquire a minus
sign under a gauge transformation (Proposition
\ref{gauge_symmetry}).
\newline
Applying the streaming equation first to $\langle q^{p-1}_{12}\rangle_t$
gives
\be
\frac{d \langle q_{12}^{p-1}\rangle_t}{dt} 
= \langle q_{12}^{p-1} \Big( q_{12}^{p-1} - 4 q_{13}^{p-1} + 3 q_{34}^{p-1}  \Big)\rangle_t, \ee
where we have also exploited the permutation symmetry among
replicas.
As a consequence, because filled monomials do not depend on $t$ in the thermodynamic limit and in $\beta$-average, we can write $\langle q_{12}^{p-1}\rangle_t \sim \langle q_{12}^p \rangle t + O(q^{2p})$, such that the first terms of the cavity function read off as
\be
\Psi(t=\beta^2) = \frac{\beta^2}{4}p - \frac{\beta^4}{8}p \langle q_{12}^{2(p-1)} \rangle + O(q_{12}^{2(p-1)}).
\ee
Hence, we can write the representation of the free energy in terms of irreducible overlap correlation functions as stated in the next
\begin{proposition}
The leading terms of the free energy of the p-spin glass model are given by the following expression in terms of overlap correlation functions:
\be\label{ultras}
\alpha(\beta)=\ln 2 + \frac{\beta^2}{4}\Big( 1 + (p-1)\langle q^p \rangle - \frac{\beta^2}{2}p \langle q^{2(p-1)}  \rangle + O( \langle q_{12}^{2(p-1)} \rangle) \Big).\ee
\end{proposition}
Note that this expression coincides with the corresponding expression for the SK model when $p=2$, in fact in this case $\langle q^p \rangle =  \langle q^{2(p-1)}  \rangle$ and the coefficient for the second moment, i.e. $\langle q_{12}^2 \rangle$, is given by $(1-\beta^2)$, which when equal to zero, i.e. at $\beta=1$, reverses the concavity of the term, implying a second-order phase transition, so that criticality is restored, as expected.
\newline
Note further that this coincides with the expansion (\ref{confronto}) of the free energy previously obtained with the Hamilton-Jacobi technique.

\subsubsection{Locking of the order parameters}

The free energy expression above has an interesting interpretation
if we regard the pressure as a function of temperature {\em and}
of all the averages of filled overlap monomials. To emphasize this
we write in the following discussion
$\alpha(\beta,\langle \cdot \rangle)$ instead of $\alpha(\beta)$; here
$\langle \cdot \rangle$ refers to the collection of all (averages of)
filled monomials and we associate to any combination of monomials a graph where each node represents a different replica and each link corresponds to an overlap between the connected nodes/replicas \cite{irriducibile}. We will show that the total temperature derivative
of $\alpha$ equals its partial derivative; in the latter, the
graphs are taken as constant, i.e.\ their temperature dependence
is not accounted for. This is reminiscent of the situation where a
free energy is expressed as an extremum over some order
parameters, and the first order variation with temperature can be
found while keeping the order parameters constant. The result we
prove shows that the filled graphs in our framework behave
similarly to such order parameters, even though of course their
values are not determined via an extremization.

In order to prove our statement, it is convenient to work with
derivatives w.r.t.\ $\beta^2$; of course $\beta$-derivatives can
be recovered trivially by multiplying by $2\beta$. From
Theorem~\ref{baffo} we have for the
pressure the expression \be \alpha(\beta,\langle.\rangle)=\ln
2+\Psi(\beta,t=\beta^2)-\frac{\beta}{2}(p-1)\partial_{\beta}\alpha(\beta). \label{alpha_graphs} \ee Its total derivative
with respect to $\beta^2$ is:
\be \label{eq:111}
 \frac{d}{d\beta^2}\alpha(\beta,\langle.\rangle)
=\partial_{\beta^2}\alpha(\beta,\langle.\rangle)+\sum_{\langle.\rangle}
\frac{\partial\alpha(\beta,\langle.\rangle)}{\partial
\langle.\rangle} \frac{\partial \langle.\rangle}{\partial\beta^2},
\ee
where the sum $\sum_{\langle.\rangle}$ runs over all filled graphs. Of
course, we already know the value of this total derivative  as it is
proportional to the internal energy:
\be \frac{d}{d\beta^2}\alpha(\beta,\langle . \rangle)=
\frac{1}{2\beta}\frac{d \alpha(\beta)}{d
\beta}=\frac{1}{4}(1- \langle q_{12}^p\rangle). \ee
But we can also calculate the partial $\beta^2$ derivative: from
(\ref{alpha_graphs}),
\be
\partial_{\beta^2}\alpha(\beta,\langle . \rangle)=
\partial_{\beta^2}\Psi(\beta,t=\beta^2)-
\frac{(p-1)}{4}(1- \langle q_{12}^p
 \rangle).
\label{alpha_partial_aux}\ee
To understand how to calculate the partial derivative of the
cavity function, where all filled monomials are held constant, recall
the expression (\ref{ultras}). We need to substitute
$t=\beta^2$ there as we are concerned with
$\Psi(\beta,t=\beta^2)$. The explicit dependence on $\beta^2$ of
the result then comes only from the prefactors of the filled
graphs, i.e.\ from the original $t$-dependence of the cavity
function. The latter is already known (see
Eq.~(\ref{deripsi_again})), and so we get
\be
\partial_{\beta^2}\Psi(\beta,t=\beta^2)=
\partial_t\Psi(\beta,t)\mid_{t=\beta^2}=
\frac{p}{4}(1- \langle q_{12}^{p-1} \rangle_{t=\beta^2})=\frac{p}{4}(1-
\langle q_{12}^p\rangle),\ee
where in the last step we have exploited Theorem~\ref{benson}.
Inserting the previous expression into (\ref{alpha_partial_aux}) shows that the total and
partial derivatives of $\alpha$ are indeed the same, as claimed:
\be \frac{d}{d\beta^2}\alpha(\beta, \langle .\rangle)=
\partial_{\beta^2}\alpha(\beta,\langle .\rangle).
\ee
As a consequence, the second term in the r.h.s. of Eq. (\ref{eq:111}) has to be identically
zero:
\be\label{criterion} \sum_{\langle .\rangle}
\frac{\partial\alpha(\beta,\langle .\rangle)}{\partial \langle
.\rangle} \frac{\partial \langle .\rangle}{\partial\beta^2}=0. \ee
We will see in section \ref{sec54} how this relates to the well-known polynomial identities
that we revise in the next section.

\subsection{A digression on Ghirlanda-Guerra and Aizenman-Contucci identities}\label{OCG}

In the $p=2$ case (namely the paradigmatic SK model \cite{note2})
Parisi went beyond the solution for the free energy and gave an
ansatz about the pure states of the model as well, prescribing the
so-called ultrametric or hierarchical organization of the phases
(see \cite{mpv} and references therein). From a rigorous point of
view, the closest the community has so far got to ultrametricity
is in the proof of identities constraining the probability
distribution of the overlaps, namely the Aizenman-Contucci (AC)
and the Ghirlanda-Guerra identities (GG) (see \cite{aizcont, gg}
respectively). These are consistent with, but weaker than,
Parisi's ultrametric structure, despite recent fundamental step forward have been achieved \cite{panchenko}.

In a nutshell, here, we summarize what the GG or AC identities state for the $p=2$ case.
Consider the overlaps among $s$ replicas. Add one replica $s+1$;
then the overlap $q_{a,s+1}$ between one of the first $s$ replicas
(say $a$) and the added replica $s+1$ is either independent of all
other overlaps, or it is identical to one of the overlaps
$q_{ab}$, with $b$ ranging across the first $s$ replicas except
$a$. Each of these cases has equal probability $s^{-1}$.

This property is very close to the relation obtained within the
Parisi picture:
 Integrating over $q_{23}$ in this equation, the joint probability
distribution for the overlaps $q_{12}$ and $q_{13}$ 
corresponding to the case
$s=2$, $a=1$ above becomes \be
P(q_{12},q_{13})=P(q_{12})\left[\frac{1}{2} \delta(q_{12}-q_{13})
+ \frac{1}{2} P(q_{13}) \right] \label{GG1_s2_Parisi} \ee where
$P(.)$ is the probability distribution of the overlap between any
two replicas. Dividing by $P(q_{12})$ gives the conditional
probability $P(q_{13}|q_{12})$, and the formula above then says
precisely that the two overlaps are independent with probability
one half and identical with the same probability. Even when we
consider two overlaps between two distinct pairs of replicas the
correlation remains strong; in fact, still following Parisi \be
P(q_{12},q_{34})= \frac{2}{3}P(q_{12})P(q_{34}) +
\frac{1}{3}P(q_{12})\delta(q_{12}-q_{34}). \label{GG2_s2_Parisi}
\ee

\subsection{Zero average polynomials at even $p$} \label{sec54}

Let us now see how to prove these properties in p-spin glasses (or at least the
equality of the second moments of the relevant distributions)
 following Ghirlanda and Guerra argument  \cite{gg}. Denote by $e(\sigma)=H_N(\sigma)/N$ the energy density;
the dependence on $N$ will be left implicit below. This quantity
is self-averaging:
\be\label{ghirla}
\lim_{N \rightarrow
\infty}(\langle e(\sigma)^2 \rangle - \langle e(\sigma) \rangle^2)=0.
\ee
Let us sketch an euristic proof of (\ref{ghirla}):
\be \label{376}
\langle e(\sigma)^2 \rangle - \langle e(\sigma) \rangle^2 =
\mathbb{E} \omega(e(\sigma)^2)-[\mathbb{E}\omega(e(\sigma))]^2=
 \mathbb{E}[\omega(e(\sigma)^2) - \omega^2(e(\sigma))] +
[\mathbb{E}\omega^2(e(\sigma)) - (\mathbb{E}\omega(e(\sigma)))^2].
\ee
The second term is the variance with the disorder of the Boltzmann
average of the energy density and, as $N \rightarrow \infty$, it goes
to zero. The first term is equal to
$-N^{-1}\partial_{\beta}\mathbb{E}\omega(e(\sigma))$ and, since $\mathbb{E}\omega(e(\sigma))$ is
finite, the prefactor $N^{-1}$ forces also this contribution to go
to zero as $N \rightarrow \infty$.
A rigorous proof for the $p=2$ case can be found in \cite{contuccigiardina} and in \cite{talagrand} for a generic even $p$.

The property (\ref{ghirla}) is fundamental because it implies, for
any function $F_s$ of overlaps among $s$ replicas,
\be\label{inequa} \lim_{N\rightarrow \infty}(\langle
e(\sigma^a)F_s\rangle - \langle e(\sigma) \rangle\langle F_s
\rangle) = 0, \ee where by $e(\sigma^a)$ we mean $e(\sigma)$
calculated on replica $a$, taken to be one of the replicas that
appear in $F_s$. Equation (\ref{inequa}) can be obtained easily
from the Schwartz inequality:
\begin{eqnarray}
&& \lim_{N \rightarrow \infty} (\langle e(\sigma^a)F_s\rangle
- \langle e(\sigma)\rangle\langle F_s\rangle)^2 = \\
&& \lim_{N \rightarrow \infty}\langle (e(\sigma^a)-\langle
e(\sigma)\rangle)F_s\rangle^2 \leq \\
&& \lim_{N \rightarrow \infty}\langle (e(\sigma^a) -\langle
e(\sigma)\rangle)^2\rangle\langle F_s^2 \rangle = 0
\end{eqnarray}
The first term in (\ref{inequa}) can be
evaluated again using Gaussian integration by parts:
%
\be\label{laprima}
\langle e(\sigma^a)F_s\rangle = - \sqrt{\frac{p!}{2N^{p-1}}} \sum_{i_1<...<i_p} \mathbb{E}
J_{i_1,...,i_p} \Omega(F_s \sigma_{i_1}^a ... \sigma_{i_p}^a) \ =  -\frac{\beta}{2}\langle F_s(\sum_{1\leq b\leq s} q_{ab}^p - s
q_{a,s+1}^p)\rangle,
\ee
while the second term is simply
\be\label{laseconda} \langle e(\sigma) \rangle \langle F_s \rangle
= -\frac{\beta}{2}(1 - \langle q_{12}^p \rangle)\langle F_s
\rangle. \ee
Combining equations (\ref{laprima}) and (\ref{laseconda}) we
obtain the first type of GG relation
\be\label{identita1} \lim_{N \rightarrow \infty} \langle F_s
\Big(\sum_{1\leq b \leq s}q_{ab}^p - s q^p_{a,s+1} - (1- \langle
q_{12}^p \rangle)\Big)\rangle = 0. \ee
Since $F_s$ is a generic function, this result implies~\cite{gg}
that, conditionally on all the overlaps $q_{cd}$ with $1\leq
c<d\leq s$,
\be
\langle q_{a,s+1}^p \rangle = \frac{1}{s} \langle q_{12}^p \rangle
+ \frac{1}{s}\sum_{1\leq b\leq s, b\neq a}q_{ab}^p.
\label{2nd_moment_conditional} \ee This is consistent with our
description above of the physical content of the GG relations; the
particular example $s=2$, $a=1$ corresponds to the second moment of
(\ref{GG1_s2_Parisi}).

In the same way  it is possible to derive a constraint for
averages involving $s+2$ replicas by using \be
\mathbb{E}\Omega(e(\sigma))\Omega(F_s)
-\mathbb{E}\Omega(e(\sigma))\mathbb{E}\Omega(F_s)=0,
\label{2nd_selfaveraging} \ee which is based on the vanishing of
the second term of eq.~(\ref{376}). One obtains the second type of
GG identity,
\be\label{simil2} \langle F_s \Big( \sum_{1 \leq b \leq
s}q_{b,s+1}^p + \langle q_{12}^p \rangle -
(s+1)q_{s+1,s+2}^p\Big)\rangle=0. \ee
Again, invoking the arbitrariness of $F_s$, this tells us that
conditional on the overlaps among the first $s$ replicas
\begin{eqnarray}
\langle q^p_{s+1,s+2} \rangle &=& \frac{1}{s+1}\sum_{1 \leq b \leq
s}\langle q_{b,s+1}^p\rangle + \frac{1}{s+1} \langle q_{12}^p
\rangle \nonumber
\\
&=& \frac{2}{s+1}\langle q_{12}^p \rangle + \frac{2}{s(s+1)}
\sum_{1\leq a<b\leq s}q_{ab}^p,
\end{eqnarray}
where the second equation follows by inserting
(\ref{2nd_moment_conditional}). The specific case $s=2$
corresponds to the second moment of (\ref{GG2_s2_Parisi}) as expected.

Finally, subtracting (\ref{2nd_selfaveraging}) from
(\ref{inequa}), which is equivalent to exploiting the vanishing of
the first term in eq.~(\ref{376}), leads to the self-averaging
relation
\be
\mathbb{E}[\Omega(e(\sigma^a)F_s)-\Omega(e(\sigma^a))\Omega(F_s)]=0,
\ee from which it is possible to obtain, again for some fixed
$1\leq a \leq s$, \be \langle F_s \Big( \sum_{1 \leq b \leq s, b
\neq a}q_{ab}^p - s q_{a,s+1}^p - \sum_{1\leq b\leq s} q_{b,s+1}^p
+ (s+1)q_{s+1,s+2}^2 \Big) \rangle. \ee
Summing over $a$ and dividing by two, this last
relation becomes
\be\label{simil3} \langle F_s \Big( \sum_{1 \leq a < b \leq
s}q_{ab}^p -s \sum_{1 \leq a \leq s}q_{a,s+1}^p +
\frac{s(s+1)}{2}q_{s+1,s+2}^p \Big)\rangle = 0, \ee which is the
general form of the AC relations.
It is interesting to note that the l.h.s.\ of
eq.~(\ref{simil3}) equals $2N\beta\partial_{\beta}\langle
F_s\rangle$, as one verifies by direct calculation: As the
$\beta$-derivative must be $O(1)$ we can then directly argue that
(\ref{simil3}) vanishes for large $N$, and does so generically as
$1/N$.

Moving on to concrete examples, the most famous GG relations are
those obtained from
$F_s=q_{12}^p$, where the exponent $p$ makes us focus on the energy term of the p-spin model. They are typically written in the form
\begin{eqnarray}
\langle q_{12}^p q_{13}^p \rangle = \frac{1}{2} \langle
q_{12}^{2p}\rangle +
\frac{1}{2}\langle q_{12}^p \rangle^2 \\
\langle q_{12}^p q_{34}^p \rangle = \frac{1}{3} \langle
q_{12}^{2p}\rangle + \frac{2}{3}\langle q_{12}^p \rangle^2.
\end{eqnarray}
Eliminating $\langle q_{12}^p\rangle^2$, we get, as expected,
the AC relation for $F_s=q_{12}^p$ \be\label{ACrelation1} \langle
q_{12}^{2p} \rangle - 4 \langle q_{12}^p q_{13}^p \rangle + 3 \langle
q_{12}^p q_{34}^p
 \rangle = 0.\ee

\bigskip

\subsubsection{Overlap constraint generators}

We now show that within our smooth cavity field framework these
relations can be obtained very simply from the stochastic
stability of filled monomials (Theorem~\ref{richard}).
Specifically, we claim that the AC identities follow from the
$t$-independence that obtains for averages of such monomials when
$N\to\infty$, and specifically from the vanishing of the
$t$-derivative at $t=\beta^2$: if $F_s$ is a filled monomial, then
\be \lim_{N\rightarrow\infty}\partial_t \langle F_s
\rangle|_{t=\beta^2}=0.\label{generator} \ee This property, for
generic $t$, has already been used in our smooth cavity expression, where we did not evaluate the streaming of
filled graphs like $q^p_{12}$ because they are independent of $t$.
\newline
To see that we can also generate constraints for the overlaps, we
combine $t$-independence with the fact that for $t=\beta^2$, by
Theorem~\ref{benson}, the perturbed Boltzmann state effectively
reverts to the unperturbed state of an enlarged system.
Explicitly, we have by evaluating the $t$-derivative in
(\ref{generator}) using the streaming equation
(Theorem~\ref{streaming_theorem}):
\be \lim_{N \rightarrow \infty}\langle F_s \Big( \sum_{1 \leq a <
b \leq s}q^{p-1}_{ab} - s \sum_{1 \leq a \leq s}q^{p-1}_{a,s+1} +
\frac{s(s+1)}{2}q^{p-1}_{s+1,s+2}\Big)\rangle_{t=\beta^2}=0. \ee
Now, given that $F_s$ is filled, all the terms here are fillable.
Because we are evaluating at $t=\beta^2$, then from
Theorem~\ref{benson} their averages reduce to unperturbed averages
of the corresponding filled expressions. Filling in this case just
means squaring all the overlaps inside the brackets, and so we get
directly
\be \lim_{N \rightarrow \infty}\langle F_s \Big( \sum_{1 \leq a <
b \leq s}q^p_{ab} - s \sum_{1 \leq a \leq s}q^p_{a,s+1} +
\frac{s(s+1)}{2}q^p_{s+1,s+2}\Big)\rangle=0. \ee
This is nothing but the general AC relation~(\ref{simil3}), as
claimed.
\newline
From the streaming of the simplest filled monomial (i.e.\ $\langle
q_{12}^p \rangle$) we find the first AC relation \bc\be
\lim_{N\rightarrow\infty} \partial_t \langle q^p_{12}
\rangle_{t=\beta^2} = \lim_{N\rightarrow\infty} \langle
q_{12}^{2p}-4 q^{p}_{12}q_{23}^p+
3q_{12}^p q_{34}^p \rangle=0, \ee\ec
which denotes overall a perfect agreement among results from our approach and previous knowledge on p-spin models.

\section{Conclusions}

In recent years spin-glasses have attracted a growing interest raised as, day after day,
these systems are becoming the bricks for building models to describe behavior of complex systems,
ranging from biology  to economics.
As a consequence, there is a need for stronger and stronger analytical methods possibly related with the numerical and experimental findings.
This paper wes written with the intention of offering a detailed analysis of a well-known model, namely the p-spin-glass, through two recent methods: the Hamilton-Jacobi technique and the smooth cavity expression. We first provide a picture of the behavior of the p-spin-glass, from a perspective which is intermediate between that of the pure theoretical physicist and that of the rigorous mathematician, hoping to help in bridging the gap between these two approaches. Then, we explain the methods and use them with abundance of details, so to allow the reader to learn them. Indeed, our focus is more on techniques and on their versatility rather than on results themselves, which are mostly already known \cite{gardner,mezardgross,talagrand}.

Summarizing, after a streamlined introduction to the basic properties of the model (expression for the internal energy and convergence of the infinite volume limit), within the Hamilton-Jacobi technique, we obtained analytical expressions for both the RS and the 1-RSB free energies and we gain a clear mathematical control of the underlying physical assumptions. Within the smooth-cavity method, we showed how to build the expression of the free energy through overlap correlation functions and we analyzed the polynomial identities, that always develop in frustrated systems, derived as consequences of the stability of the measure $\lim_{N\to \infty}\sum_{\sigma}$ with respect  to small (negligible) random stochastic perturbations.

It is interesting to note that in the mean field techniques we developed, there is a certain degree of complementary between the two approaches as within the former we use a trial overlap coupled with a single particle free spin, while in the latter we use $p -1$ spins for the cavity: both methods essentially work by reducing the problem to a single-body one, the Hamilton-Jacobi in a direct way, the smooth-cavity in a complementary way.

Further investigations should be addressed to the study of the diluted frustrated p-spin model and its relation with K-satisfiability problems and $P/NP$ completeness.

\subsection*{Acknowledgements}
This research belongs to the strategy founded by FIRB project $RBFR08EKEV$ and by Sapienza University.
AB is grateful to Francesco Guerra and Peter Sollich for useful discussions on stochastic stability and random energy models.

\appendix
\section{Derivatives of the generalized partition function respect to the interpolating parameters}

We show here the derivation of expressions (\ref{eq:dtalfa}, \ref{eq:dxalfa}).
Let's start with the first one:
\be
\partial_t \tilde{\alpha}_N = \frac{1}{N} \easp_0 Z_0^{-1} \partial_t Z_0
\ee
It's easy to see that
\be
\label{eq:deriter}
Z_a^{-1} \partial_t Z_a = \easp_{a+1} f_{a+1} Z_{a+1}^{-1} \partial_t Z_{a+1}
\ee
so that
\be
\label{eq:z0zk}
\easp_0 Z_0^{-1} \partial_t Z_0 = \easp_0 \easp_{1}...\easp_K f_1...f_K Z_K^{-1} \partial_t Z_K \equiv \easp  f_1...f_K Z_K^{-1} \partial_t Z_K
\ee
and, remembering that $Z_K \equiv \tilde{Z}_N$,
\be
\easp_0 Z_0^{-1} \partial_t Z_0 =
\frac{1}{2 \sqrt{t}} \sqrt{ \frac{t p!}{2 N} } \sum_{1\leq i_1<...<i_p\leq N} \easp(f_1...f_K J_{i_1...i_p} \tilde{\omega}(\sigma_{i_1}...\sigma_{i_p})).
\ee
Integrating by parts the $J_{i_1...i_p}$ inside the expectation becomes a derivative:
\be
\frac{1}{2 \sqrt{t}} \sqrt{ \frac{t p!}{2 N} } \sum_{1\leq i_1<...<i_p\leq N}
\left[ \sum_{a=2}^K \easp(f_1...\partial_{J_{i_1...i_p}}f_a...f_K \tilde{\omega}(\sigma_{i_1}...\sigma_{i_p})) +
\easp(f_1... f_K \partial_{J_{i_1...i_p}} \tilde{\omega}(\sigma_{i_1}...\sigma_{i_p})) \right].
\ee
We now proceed by computing separately the two addends in the square brackets.
for the second term we easily find
\be
\partial_{J_{i_1...i_p}} \tilde{\omega}(\sigma_{i_1}...\sigma_{i_p})) = \sqrt{ \frac{tp!}{2N^{p-1}} } (1- \tilde{\omega}^2(\sigma_{i_1}...\sigma_{i_p}) )
\ee
while for the first term we have
\be
\partial{J_{i_1...i_p} } f_a = m_a f_a Z_a^{-1} \partial{J_{i_1...i_p} }  Z_a - m_a f_a E_a (f_a Z_a^{-1} \partial{J_{i_1...i_p} }  Z_a).
\ee
Now, using the analogous of (\ref{eq:deriter}), one has
\bea
Z_a^{-1} \partial{J_{i_1...i_p} }  Z_a & = & \easp_{a+1}...\easp_K (f_{a+1}...f_K Z_K^{-1} \partial{J_{i_1...i_p} }  Z_K) \\
&&  = \sqrt{ \frac{tp!}{2N^{p-1}} } \omega_a (\sigma_{i_1}...\sigma_{i_p})
\eea
and then for $a\geq2$ (remind that $f_1=1$ so its derivative is zero)
\be
\partial{J_{i_1...i_p} }  f_a  = \sqrt{ \frac{tp!}{2N^{p-1}} } m_a f_a (\omega_a (\sigma_{i_1}...\sigma_{i_p}) - \omega_{a-1} (\sigma_{i_1}...\sigma_{i_p}) ).
\ee
Putting together the terms computed we find
\be
\begin{split}
\easp_0 Z_0^{-1} \partial_t Z_0 = &
\frac{1}{4} \frac{p!}{N^{p-1}} \sum_{1\leq i_1<...<i_p\leq N}
\bigg[ \sum_{a=2}^K \easp_0...\easp_a (f_1...f_a \omega_a (\sigma_{i_1}...\sigma_{i_p}) f_{a+1}...f_K \tilde{\omega}(\sigma_{i_1}...\sigma_{i_p}) \\
& - \sum_{a=2}^K \easp_0...\easp_a (f_1...f_a \omega_a (\sigma_{i_1}...\sigma_{i_p}) f_{a+1}...f_K \tilde{\omega}(\sigma_{i_1}...\sigma_{i_p} ) \\
& +1- \easp_0...\easp_K (f_1...f_K \tilde{\omega}^2(\sigma_{i_1}...\sigma_{i_p})) \bigg]
\end{split}
\ee
and noting that in the thermodynamic limit $p! \sum_{i_1<...<i_p} \sim \sum_{i_1,...,i_p}$ we have
\be
\partial_t \tilde{\alpha}_N = \frac{1}{4} \bigg[
\sum_{a=1}^K m_a ( \langle q_{\sigma \sigma'}^p \rangle_a -\langle q_{\sigma \sigma'}^p \rangle_{a-1}) + 1- \langle q_{\sigma \sigma'}^p \rangle_K  \bigg]
\ee
from which the derivation of (\ref{eq:dtalfa}) is straightforward.

Let's now compute the derivatives of the free energy respect to the "spatial" parameters:
\be
\begin{split}
\partial_a \tilde{\alpha}_N  & =
\frac{1}{N} \easp_0 Z_0^{-1} \partial_a Z_0 \\
& = \frac{1}{N} \easp (f_1...f_K Z_K^{-1} \partial_a Z_K) \\
& = \frac{1}{N} \frac{1}{ 2 \sqrt{x_a} } q^{ \frac{p-2}{4} } \easp (f_1...f_K \sum_{i=1}^N J_i^a \tilde{\omega}(\sigma_i))
\end{split}
\ee
where we used the analogous of (\ref{eq:z0zk}).
Integrating by parts this becomes
\be
\label{eq:daalfa1}
\begin{split}
\partial_a \tilde{\alpha}_N  = &
\frac{1}{N} \easp_0 \sum_i \bigg[ \easp_1...\easp_K(\sum_{b=2}^K f_1... \partial_{J_i^a} f_b...f_K \tilde{\omega}(\sigma_i)) \\
& + \easp_1...\easp_K( f_1...f_K \partial_{J_i^a} \tilde{\omega}(\sigma_i)) \bigg].
\end{split}
\ee
The derivatives of the state and of $f_b$ respect to the random fields are respectively given by
\be
\begin{split}
\partial_{J_i^a} \tilde{\omega}(\sigma_i) & = \sqrt{x_a} q_a^{ \frac{p-2}{4} } (1 - \tilde{\omega}^2(\sigma_i) ) \\
\partial_{J_i^a} f_b & =
\begin{cases}
m_b f_b (Z_b^{-1} \partial_{J_i^a} Z_b - E_b f_b Z_b^{-1} \partial_{J_i^a} Z_b) & \mbox{if } a<b \\
m_b f_b Z_b^{-1} \partial_{J_i^a} Z_b & \mbox{if } a=b \\
0 & \mbox{if } a>b.
\end{cases}
\end{split}
\ee
Using again the iterative derivation formula we have
\be
\begin{split}
Z_b^{-1} \partial_{J_i^a} Z_b & = E_{b+1}(f_{b+1} Z_{b+1}^{-1} \partial_{J_i^a} Z_{b+1}) \\
& = E_{b+1}...E_K (f_{b+1}...f_K Z_{K}^{-1} \partial_{J_i^a} Z_{K}) \\
& = \sqrt{x_a} q_a^{ \frac{p-2}{4} }E_{b+1}...E_K (f_{b+1}...f_K \tilde{\omega}(\sigma_i) ) \\
& =  \sqrt{x_a} q_a^{ \frac{p-2}{4} } \omega_b(\sigma_i)
\end{split}
\ee
so that
\be
\partial_{J_i^a} f_b =
\begin{cases}
\sqrt{x_a} q_a^{ \frac{p-2}{4} } m_b f_b ( \omega_b(\sigma_i) - \omega_{b-1} (\sigma_i) ) & \mbox{if } a<b \\
\sqrt{x_a} q_a^{ \frac{p-2}{4} } m_b f_b \omega_b(\sigma_i)  & \mbox{if } a=b \\
0  & \mbox{if } a>b.
\end{cases}
\ee
Substituting these in the (\ref{eq:daalfa1}), the first term inside the square brackets becomes
\be
\begin{split}
\easp_1...\easp_K (\sum_{b=2}^K f_1... \partial_{J_i^a} f_b...f_K \tilde{\omega}(\sigma_i)) =
& \sqrt{x_a} q_a^{ \frac{p-2}{4} } \big[ m_a \easp_1... \easp_a (f_1...f_a \omega_a^2(\sigma_i)) \\
& + \sum_{b=a+1}^K m_b \easp_1... \easp_{b} (f_1...f_{b} \omega_b^2(\sigma_i)) \\
& - \sum_{b=a+1}^K m_b \easp_1... \easp_{b-1} (f_1...f_{b-1} \omega_b^2(\sigma_i)) \big]
\end{split}
\ee
and we find
\be
\partial_a \tilde{\alpha}_N = \frac{1}{2} q_a^{\frac{p-2}{2}}
\big[m_a \langle q_{\sigma \sigma'} \rangle_a
+ \sum_{b>a} m_b (\langle q_{\sigma \sigma'} \rangle_{b} -\langle q_{\sigma \sigma'} \rangle_{b-1}) + 1 - \langle q_{\sigma \sigma'} \rangle_K \big]
\ee
from which, after some manipulations, one can easily obtain the (\ref{eq:dxalfa}).



\begin{thebibliography}{1}
\bibitem{note2} Among the several many body theories developed in statistical mechanics along the years, the mean field two body ones are the most welcome as the Hamiltonian -even though no longer in classical space-time sense- are still quadratic forms such that linear response for the forces is still kept. However, especially in disordered system, many real features of glass forming dynamics
seem better reproduced by violation of the linear response  and in this sense by p-spin models.
\end{thebibliography}
\end{document}